\documentclass[letterpaper]{article}
\usepackage{aaai2026}
\usepackage[T1]{fontenc}
\usepackage{times}
\usepackage{helvet}
\usepackage{courier}
\usepackage[hyphens]{url}
\usepackage{graphicx}
\usepackage{natbib}
\usepackage{caption}
\usepackage{tikz}
\usepackage{array}
\usepackage{booktabs}
\usepackage{amsfonts}
\usepackage{multirow}
\usepackage{algorithm}
\usepackage{algorithmic}
\usepackage{booktabs}
 \newtheorem{definition}{Definition}
  \newtheorem{remark}{Remark}
\newtheorem{proof}{Proof}
\newtheorem{proof_sketch}{Proof Sketch}
\newtheorem{lemma}{Lemma}

\usepackage{newfloat}
\usepackage{listings}
\usepackage{amsmath}
\DeclareCaptionStyle{ruled}{labelfont=normalfont,labelsep=colon,strut=off}
\floatstyle{ruled}
\newfloat{listing}{tb}{lst}{}
\floatname{listing}{Listing}
\title{Level-k Distinguishable Mechanisms for Evaluating Bounded Rationality in LLMs}
\author{
    Binchi Zhang\textsuperscript{\rm 1},
    Atrisha Sarkar\textsuperscript{\rm 2}
}
\affiliations{
    \textsuperscript{\rm 1}University of Western Ontario\\
    \textsuperscript{\rm 2}University of Western Ontario\\
    bzhan484@uwo.ca, atrisha.sarkar@uwo.ca
}

\begin{document}
\maketitle

\begin{abstract}
Strategic depth of reasoning is essential for human interaction of Large Language Models (LLMs) operating in boundedly rational environments. However, existing evaluations are primarily based on canonical games prevalent in pretraining corpora, making it difficult to disentangle true strategic reasoning from memorisation. To address this, we formalise a necessary \emph{level-$K$ distinguishability} condition for strategic depth inference and construct a suite of novel game structures that meet this standard.  
Using these games, we evaluate strategic depth in LLMs from both the Chain-of-Thought tokens and actual actions under recursive reasoning and an inductive trace of opponent game-play data. Across experimental trials spanning four LLMs, four game structures, and ten levels of iterated reasoning, we find that models maintain accurate strategic depth under recursive reasoning, with strong internal consistency between stated reasoning and actions at every level. Errors arise from using the wrong number of iterated depth of reasoning steps, not from computing best responses incorrectly. However, inductive inference from opponent play degrades accuracy sharply and unevenly across games, and explicit strategic mentalizing in the chain of thought substantially improves overall performance.

\end{abstract}
\section{Introduction}
In the realm of strategic reasoning, behavioural game theory models agents' deviations from the Nash equilibrium of a game \cite{camerer2004behavioural}. The level-$k$ family of models (including Cognitive Hierarchy models) provides a flexible meta-model to characterize such deviations from equilibrium by identifying the depth of iterated reasoning an agent employs to model and rationalize the behaviour of their interacting partners \cite{wright2017predicting}. Studies on humans have identified a natural limit to human capacity along this dimension \cite{coricelli2009neural}, with heterogeneity arising from game structure and incentives \cite{georganas2015persistence}. Naturally, when LLMs are deployed for everyday tasks that involve reasoning about the behaviour and strategies of interacting agents (some of whom may be human), it is important to know whether an analogous natural limit to this depth exists for Large Language Models, for several reasons. First, LLMs are increasingly used to simulate human behaviour. In domains where it is important to model other humans as strategic agents, such as economic situations \cite{mallard2012modelling} and autonomous vehicle interactions \cite{sarkar2022generalized}, LLMs are being used as a strategic decision layer \citep{horton2023large, fu2024drive}, and as such, they need the capacity to infer strategic depth and respond just as humans do in such interactions. Second, identifying the strategic depth of reasoning is crucial for AI safety: an agent with a lower level of strategic sophistication can be exploited by one with a higher level \cite{alon2026}. In interactions with humans, who have a limit to their own depth, it is important to identify and calibrate the strategic depth of LLMs to prevent such exploitation.

Existing methods for inferring the strategic depth of reasoning in LLMs from a behavioural perspective, that is, through observed actions, use standard canonical games, including coordination and signalling games \cite{jia2026llm}, the Keynesian beauty contest, and negotiation games \cite{zhang2025k, trencsenyi2025approximating}, among others. However, the validity of these instruments is open to question. First, these games were originally designed as instruments to infer human, not artificial, strategic depth of reasoning, and are therefore designed to be simple enough for human comprehension in an experimental setting. Results from such gameplay are widely reported across decades of literature on behavioural game theory, which we can reasonably assume is present in internet-scale training corpora. Second, many of the game structures used for inferring strategic depth lack the capacity to infer reasoning up to an arbitrary depth of iteration. This is because inferring latent attributes---such as utilities, beliefs, and reasoning---from action traces alone runs into the well-known unidentifiability problem: multiple latent states can produce the same behaviour \cite{ng2000algorithms}. In many canonical games, a given level of strategic reasoning does not correspond to a unique observable action, and therefore suffers from the same unidentifiability issue. To improve the validity of such game-theoretic instruments, there is thus a need to design novel game structures (or mechanisms) that satisfy a one-to-one correspondence between actions and depth of iterated reasoning; we refer to this as the \emph{level-$k$ distinguishability} condition.

In a multiagent environment, designing games that elicit a latent and private attribute falls within the domain of mechanism design \cite{duetting2024mechanism, borgers2015introduction}. However, unlike mechanism design, where the goal is to design mechanisms whose solutions are compatible with the private utilities of the players and, therefore, uniquely identifiable, our goal is to design games whose solutions are compatible with and unique to a particular strategic depth of a player. This difference between eliciting latent utilities and eliciting latent reasoning means that we cannot rely on the convenience of the revelation principle \cite{kephart2016revelation} to narrow the class of games (such as auctions) used as identifiability instruments.\\
In this paper, we resolve the above set of identified problems through the following contributions.
\begin{itemize}
  \item We construct a novel game structure and adapt three canonical games to satisfy the \emph{level-$K$ distinguishability} property, using them as instruments to infer the strategic depth of reasoning in LLMs. We use these games to infer the strategic depth of four LLMs up to a previously untested depth of 10.
  \item We evaluate the consistency between the models' iterated depth of reasoning as reflected in their intermediate chain-of-thought inference and in their resulting strategies.
  \item For modelling other players' behaviour, we compare the effect of recursive versus inductive reasoning on the models' level of strategic depth.
\end{itemize}

\section{Related Work}
\label{sec:related-work}

\paragraph{Depth of reasoning and its identification.}
Level-$k$ and cognitive-hierarchy models assign each player an integer depth of iterated best response anchored at a non-strategic $L_0$ type \citep{nagel1995unraveling,costa2001cognition,camerer2004cognitive}, and fitting these types to observed actions is standard practice \citep{crawford2013structural}. Machine learning methods have been used to estimate the strategic depth as a parameter from human game-play data \cite{hartford2016deep}. In contrast to that estimation problem, we focus on the \emph{identification} problem by designing games for automated agents that distinguish each depth of reasoning. A work on the identification problem of higher order rationality is \citet{kneeland2015identifying}, which uses ring-network games for the purpose. However, unlike our focus on the level-$k$ solution concept, they focus on identification of strategies under different orders of rationalizability. That work and the game structures that follow it, such as \citet{cerigioni2019higher}, also target the depths human subjects reach, at most about four, and none states a condition for distinguishability beyond that range. We also improve upon the identifiability aspect through recursive and inductive conditioning of the information about other agents', so exactly one iterated depth of reasoning is correct and the action is identifiable against it alone.

\paragraph{Level-$k$ evaluation of LLMs.}
Many benchmarks now evaluate LLMs by having them play games \citep{duan2024gtbench,wang2024tmgbench,huang2025gama}. They report payoffs, or the share of actions matching a solution concept. However, such games cannot identify the strategic depth, because a model can reach the right action without iterative reasoning. The games themselves are also often structurally simple enough, thereby, limiting what the strategies can reveal. \citet{wang2024tmgbench} scores whether the model identifies the Nash equilibria of $2\times2$ games, which are either solved within two rounds of iterated dominance or not dominance-solvable at all. \citet{huang2025gama} includes a standard level-$k$ instrument, Guess 2/3 of the Average, but scores it by distance from equilibrium, not the strategic depth itself. Other studies apply the level-$k$ model directly. \citet{zhang2025k} prompts models to reason recursively in guessing, auction, and negotiation games. The negotiation game it adapts collapses depth 2 onto depth 0. \citet{liu2025chbench} fits level-$k$ and Poisson cognitive-hierarchy models to LLM choices, capping the maximum type at $\hat{k}=4$ where fit stops improving. \citet{fan2024rational} estimates a depth parameter under a truncated quantal-response equilibrium. Replications of the guessing games of \citet{costa2006cognition} place reasoning-tuned models near three depths, compared with fewer than two for humans \citep{kader2024emergence}. As discussed in the introduction, the above body of work does not verify that the game structure maps each depth to a distinct action over the range it probes. 

\paragraph{Chain-of-thought as evidence about reasoning.}
A second body of work asks whether stated reasoning consistently describes the computation that produced the answer. Models often act on input features their explanations never mention, or explain answers they had already settled on \citep{turpin2023language,chen2025reasoning}. It is suggested that CoT is not an account of the computation at all \citep{barez2025chain,kambhampati2026position}. In contrast to work that focus on the stronger notion of causal connection between the CoT traces and action, e.g. \citep{lanham2023measuring, kambhampati2026position}, we only measure consistency: whether the submitted action best-responds to the stated belief, and whether the two implied depths agree.

\paragraph{Contamination.}
Replacing named games with procedurally sampled payoff matrices dramatically reduces performance, indicating that much apparent strategic competence is recall rather than computation \citep{nie2026equilibrium}. A similar effect has been observed by perturbing the payoffs; models continue to play the canonical strategy after the equilibrium change \citep{georgousis2026counterfactual}. In comparison to such modifications, our approach keeps both the level-k distinguishability condition and transformation of the game structures, which is relatively non-trivial compared to random changes to the game payoffs.

\section{Methodology}
\label{sec:methodology}

The level-$k$ family of models---including its cognitive-hierarchy generalization---characterizes boundedly rational play by assigning each player a discrete cognitive type $k \in \{0, 1, 2, \dots\}$ that indexes the depth of iterated reasoning about others \citep{stahl1994experimental, stahl1995players, nagel1995unraveling, costa2001cognition}. A level-0 ($L_0$) type is \emph{non-strategic}: it forms no beliefs about other players, and its action $a^{(0)} \in S$ is specified exogenously through a non-strategic level-0 \emph{anchor} of the model. Every level-$k$ type with $k \geq 1$ believes its opponents are level-$(k-1)$ and best-responds to that belief,
\begin{equation}
    a^{(k)} \;=\; \operatorname{BR}\!\big(a^{(k-1)}\big)
    \;\equiv\; \arg\max_{s \in S}\, u_i\!\big(s,\, a^{(k-1)}\big),
    \label{eq:br-iteration}
\end{equation}
so that the latent depth $k$ is mapped to observable behavior through iterated best response. Cognitive-hierarchy models relax the degenerate belief in Eq.~\eqref{eq:br-iteration} to a distribution over all lower types \citep{camerer2004cognitive}, but retain the same structural core of a non-strategic anchor, an integer-valued latent depth, and a deterministic level-to-action map \citep{crawford2013structural}. We call the induced sequence $\{a^{(k)}\}_{k=0}^{K_{\max}}$ the \emph{iterated depth response sequence} (IDR sequence) of a game, i.e, the best-response action sequence generated by increasing the depth of iterated reasoning one level at a time.

Inferring the latent depth $k$ from an observed action is identifiable only if the level-to-action map is invertible at least up to the strategic depth we want the game instrument to reveal. Many canonical games fail this condition. For example, a game structure used in \citet{zhang2025k} to infer strategic depth is an adaptation of an item-division negotiation game. In the game, two agents should agree on how to split a pool of items of $m$ types under imperfect information. Abstracting the alternating-offer protocol in its one-shot demand form, let $q \in \mathbb{Z}_{\geq 0}^{m}$ be the pool, $w_i \in \mathbb{Z}_{\geq 0}^{m}$ agent $i$'s private utility vector, and a proposal a bundle $s_i \in B(q) \equiv \{s \in \mathbb{Z}_{\geq 0}^{m} : s \leq q\}$ claimed for oneself, with payoffs
\[
    u_i(s_i, s_{-i}) \;=\; \big(w_i^{\top} s_i\big)\, \mathbf{1}\big[s_i + s_{-i} \leq q\big].
\]
Whenever $w_i$ is strictly positive, the best response to any proposal is to claim the entire remaining $\operatorname{BR}_i(s_{-i}) = q - s_{-i}$. A level-0 agent is non-strategic, and let's say, they propose a bundle $a^{(0)} = s$. A level-1 agent best-responds with $a^{(1)} = q - s$. A level-2 agent best-responds to that with $a^{(2)} = q - (q - s) = s = a^{(0)}$. In this simple illustrative case, the level-2 proposal coincides with the level-0 proposal, and from behavior alone, a depth-0 and a depth-2 agent are therefore indistinguishable for every anchor and every utility profile.

\begin{definition}[Level-$K$ distinguishability]
\label{def:distinguishability}
Let $G$ be a game with IDR sequence $\{a^{(k)}\}_{k=0}^{K_{\max}}$, where $a^{(0)}$ is the $L_0$ anchor and $a^{(k)} = \operatorname{BR}(a^{(k-1)})$. $G$ satisfies \emph{distinguishability within depth $K_{\max}$} iff
\[
    \forall\; 0 \leq k_1 < k_2 \leq K_{\max}: \quad a^{(k_1)} \neq a^{(k_2)},
\]
i.e., $k \mapsto a^{(k)}$ is injective on $\{0, \dots, K_{\max}\}$.
\end{definition}

On a bounded integer strategy space, the iteration in Eq.~\eqref{eq:br-iteration} must eventually revisit a value, resulting either in a fixed point or a periodic cycle. Therefore, we can only infer the strategic depth of an agent through a level-k model up to depth $K_{\max}$. 

\subsection{Level-k distinguishable game structures}

\subsubsection{Ring 11--20}We first construct a novel game structure whose iterated depth response sequence is \emph{level-k distinguishable} by adapting the 11--20 money-request game \citep{arad2012eleven} into a five-player ring inspired by the ring games of \citet{kneeland2015identifying}. Players are arranged on a directed cycle with indices taken modulo 5 (Figure~\ref{fig:ring}); player $i$ chooses $s_i \in \{1, \dots, 20\}$ and receives a primary bonus of $+10$ for undercutting the left neighbor by exactly 3 and a secondary bonus of $+3$ for exceeding the right neighbor by exactly 7 (Table~\ref{tab:four-games}). A level-$k$ player best-responds to the belief that \emph{both} neighbors are level-$(k-1)$.
 
\begin{table*}[!t]
    \centering
    \small
    \setlength{\tabcolsep}{3pt}
    \caption{Summary of the game structures with level-$k$ distinguishability.}
    \label{tab:four-games}
    \begin{tabular}{l p{0.18\textwidth} p{0.16\textwidth} p{0.21\textwidth} p{0.19\textwidth}}
        \toprule
        & \textbf{11--20 (Mod.)} & \textbf{All-Pay Auction} & \textbf{Nash Demand} & \textbf{Ring 11--20} \\
        \midrule
        Players & 2 & 2 & 2 & 5 (ring) \\
        Strategy space $S$ & $\{50,\dots,70\}$ & $\{0,\dots,10\}$ & $\{1,\dots,20\}$ & $\{1,\dots,20\}$ \\
        Payoff & $u_i(s_i,s_j) = s_i$ $+\,24\,\mathbf{1}[s_i = s_j-2]$ $+\,8\,\mathbf{1}[s_i = s_j-1]$
        & $u_i(s_i,s_j) = 12\,\mathbf{1}[s_i > s_j]$ $+\,6\,\mathbf{1}[s_i = s_j] - s_i$
        & $u_i(s_i,s_j) = s_i\,\mathbf{1}[s_i+s_j \leq 20]$ $+\,0.8\,s_i\,\mathbf{1}[s_i+s_j > 20]\,\mathbf{1}[s_i < s_j]$
        & $u_i(s_i,\mathbf{s}_{-i}) = s_i$ $+\,10\,\mathbf{1}[s_i = s_{i-1}-3]$ $+\,3\,\mathbf{1}[s_i = s_{i+1}+7]$ \\
        \addlinespace[4pt]
        $a^{(0)}$ & 70 & 0 & 20 & 20 \\
        \addlinespace[2pt]
        IDR sequence & Decreasing & Increasing & Decreasing & Periodic \\
        \addlinespace[2pt]
        $K_{\max}$ & 10 & 10 & 10 & 9 \\
        \bottomrule
    \end{tabular}
\end{table*}
 
\begin{figure}[!t]
    \centering
    \begin{tikzpicture}[>=stealth,
        player/.style={circle, draw, minimum size=8mm, inner sep=1pt, font=\scriptsize}]
        \node[player] (Pi)  at (90:1.4cm)  {$i$};
        \node[player] (Pp1) at (162:1.4cm) {$i{+}1$};
        \node[player] (Pp2) at (234:1.4cm) {$i{+}2$};
        \node[player] (Pm2) at (306:1.4cm) {$i{-}2$};
        \node[player] (Pm1) at (18:1.4cm)  {$i{-}1$};
        \foreach \a/\b in {Pi/Pm1, Pp1/Pi, Pp2/Pp1, Pm1/Pm2, Pm2/Pp2}
            \draw[->, thick] (\a) to[bend left=18] (\b);
        \foreach \a/\b in {Pi/Pp1, Pp1/Pp2, Pp2/Pm2, Pm2/Pm1, Pm1/Pi}
            \draw[->, dashed] (\a) to[bend left=18] (\b);
    \end{tikzpicture}
    \caption{The Ring 11--20 game (player indices modulo 5). Solid arrows point from each player $i$ to the left neighbor $i-1$ carrying the primary $+10$ bonus ($s_i = s_{i-1} - 3$); dashed arrows point to the right neighbor $i+1$ carrying the secondary $+3$ bonus ($s_i = s_{i+1} + 7$).}
    \label{fig:ring}
\end{figure}
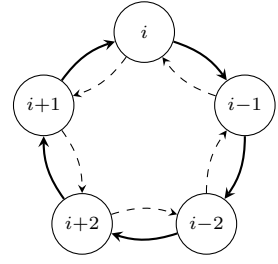
 
\begin{remark}
\label{rem:ring-periodicity}
The Ring 11--20 IDR sequence is periodic with period 10 and non-monotone \textit{level-k} distinguishability up to $K_{\max} = 9$.
\end{remark}
 
\begin{proof_sketch}
The best response to a common neighbor action $c$ has two cases: the primary-bonus action $c - 3$ when $c \geq 14$, and the secondary-bonus action $c + 7$ when $11 \leq c \leq 13$. Iterating from the anchor $a^{(0)} = 20$ generates the sequence $20, 17, 14, 11, 18, 15, 12, 19, 16, 13$, whose ten values are pairwise distinct, after which $a^{(10)} = 20$ returns to the anchor. Hence $k \mapsto a^{(k)}$ is injective through depth 9 and the sequence is periodic with period 10; the upward jumps at $a^{(4)}$ and $a^{(7)}$ give non-monotonicity.
\end{proof_sketch}
 
The structure of Ring 11--20 also has the benefit that the IDR sequences are non-monotonic, unlike the standard canonical games. This enables us to test that a model does not trivially extrapolate a strategy; for example, by subtracting a constant integer at each depth a model might spuriously generate actions that correspond to increasing depth of reasoning.
 
\subsubsection{Canonical Game Adaptations}
\label{sec:canonical-games}
We complement the Ring 11--20 game with adaptation of three canonical two-player games---the 11--20 money-request game \citep{arad2012eleven}, All-pay auction \citep{baye1996allpay}, and the Nash demand game \citep{nash1953two}. Our adaptation involves constructing the payoffs of each of these games to adhere to the \emph{level-k distinguishability} condition in Definition~\ref{def:distinguishability} over the full range of strategic depth we wish to test in our experiments. Because the 11--20 game is often used as an instrument to infer level-k reasoning, we make an additional minor modification to the strategy space and payoffs through an affine transformation of the original game, so that the standard published solution is no longer optimal in the version the models face.

\subsection{Intermediate Reasoning and Action Consistency}
\label{sec:cot-action-consistency}
 
A key component of level-$k$ family of models is a player's inference about the depth of strategic reasoning of others \cite{stahl1994experimental}. A level-$(k{+}1)$ player is, by construction, one who believes its opponent reasons at level $k$. When an LLM playing a game receives information about the opponent, the intermediate tokens of the chain of thought (CoT) contain explicit statements of the model's belief about the opponent's action, while the submitted action is a separate behavioural trace of whatever computation actually determined the choice. Given that these two signals have been observed to be logically uncorrelated in LLM CoT \cite{kambhampati2026position}, we extract a depth estimate from each and measure their consistency using three primary metrics: (i) \emph{Accuracy}, which evaluates whether the behavioural action matches the target level of the IDR sequence; (ii) \emph{Best-response consistency}, which evaluates whether the behavioural action is a best response to the model's \emph{own stated} belief (allowing for multiple best responses). Note that the two can be orthogonal---a model can be consistent to a wrong belief, or accurate despite an inconsistent narration in CoT trace; (iii)  \emph{Internal consistency}, which evaluates the relationship between the CoT and the behaviour depth estimates. If $K_L$ and $K_B$ denote the model's estimate about the opponent's strategic depth in CoT signals and the action consistent with that belief, $K_L > K_B$ indicates the model verbally performs deeper reasoning than its actions realise, while $K_B > K_L$ indicates that the action uses a reasoning level never articulated in the CoT.

\subsubsection{Recursive and Inductive Conditioning}
\label{sec:two-phase-design}
 
Since a level-$k$ reasoning agent needs information about its opponent's strategic depth in order to best respond, this information can be fed to the model in two different ways. Either by specifying the explicit level of the opponent, which requires the agent to engage in recursive reasoning by simulating the opponent's optimal strategy in an iterated way. Or by giving samples from a distribution of opponents' strategy, which requires the agent to infer the level from the behavioural trace and best respond. We construct one experimental condition for each of these two methods.
 
\textit{Recursive conditioning.} The model receives a description of the opponent that conveys the opponent's strategic depth. With this information, the model needs to extract the level-0 ($L_0$) anchor, simulate the opponent's reasoning chain, form a belief, and compute the best response. 
 
\textit{Inductive conditioning} The model receives no opponent description but it can make a tool call to query the opponent's history to obtain 30 raw action samples from the opponent's previous play, generated at a target depth unknown to the model. We sample these actions following the quantal best response of the Quantal Cognitive Hierarchy model \citep{wright2010beyond}. For an opponent of level $k \geq 1$ each response $h_t$ is drawn with
\[
    P(h_t = s) \;\propto\; \exp\!\left(\lambda \cdot \frac{u\big(s;\, a^{(k-1)}\big) - u_{\min}}{u_{\max} - u_{\min}}\right),
\]
with $\lambda \in \{1, 3, 10\}$.
 
\begin{figure*}[!htbp]
\centering
\includegraphics[width=0.78\textwidth]{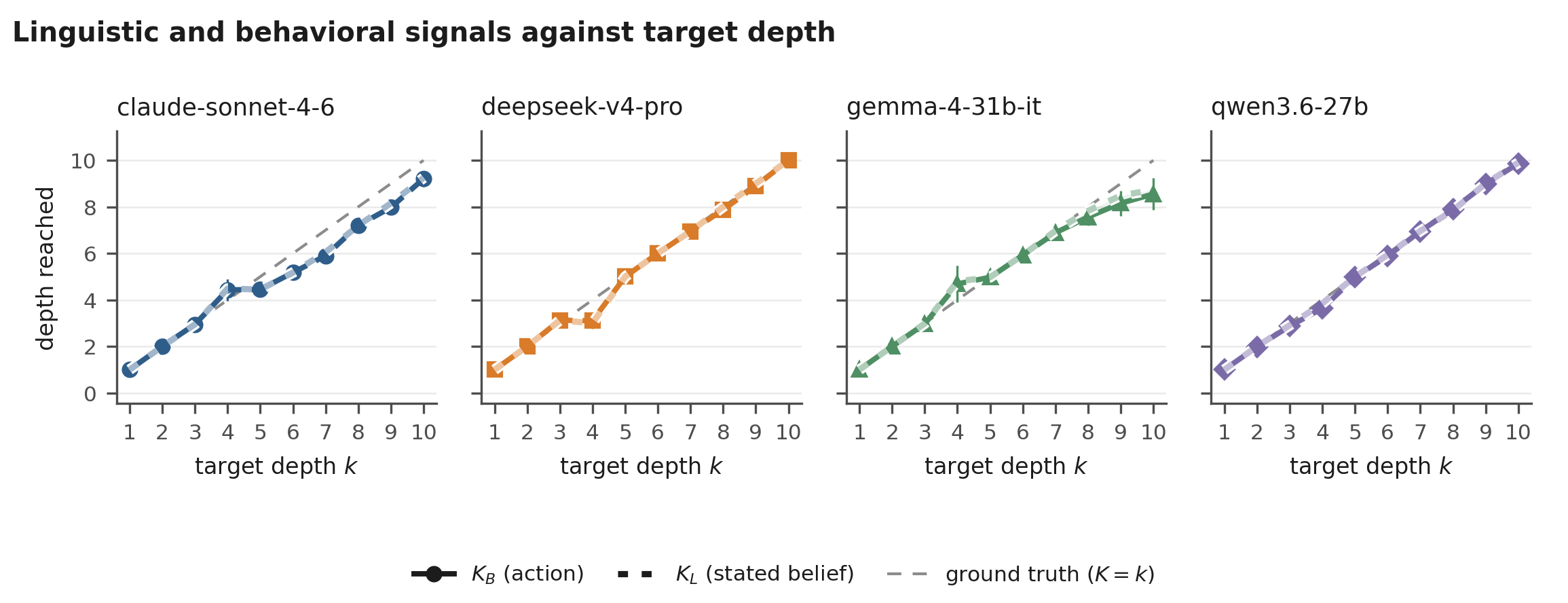}
\caption{Recursive condition: Depth limit---linguistic and behavioral signals as a function of target depth. Each panel is one model; solid lines and markers show $K_B$ (behavioral depth), dashed lines show $K_L$ (linguistic depth), and the gray dashed diagonal marks where $K$ equals target $k$.
}
\label{fig:depth-ceiling-fig}
\end{figure*}
\section{Experiment and Results}
\label{sec:models-execution}

\begin{figure*}[!htbp]
\centering
\includegraphics[width=0.78\textwidth]{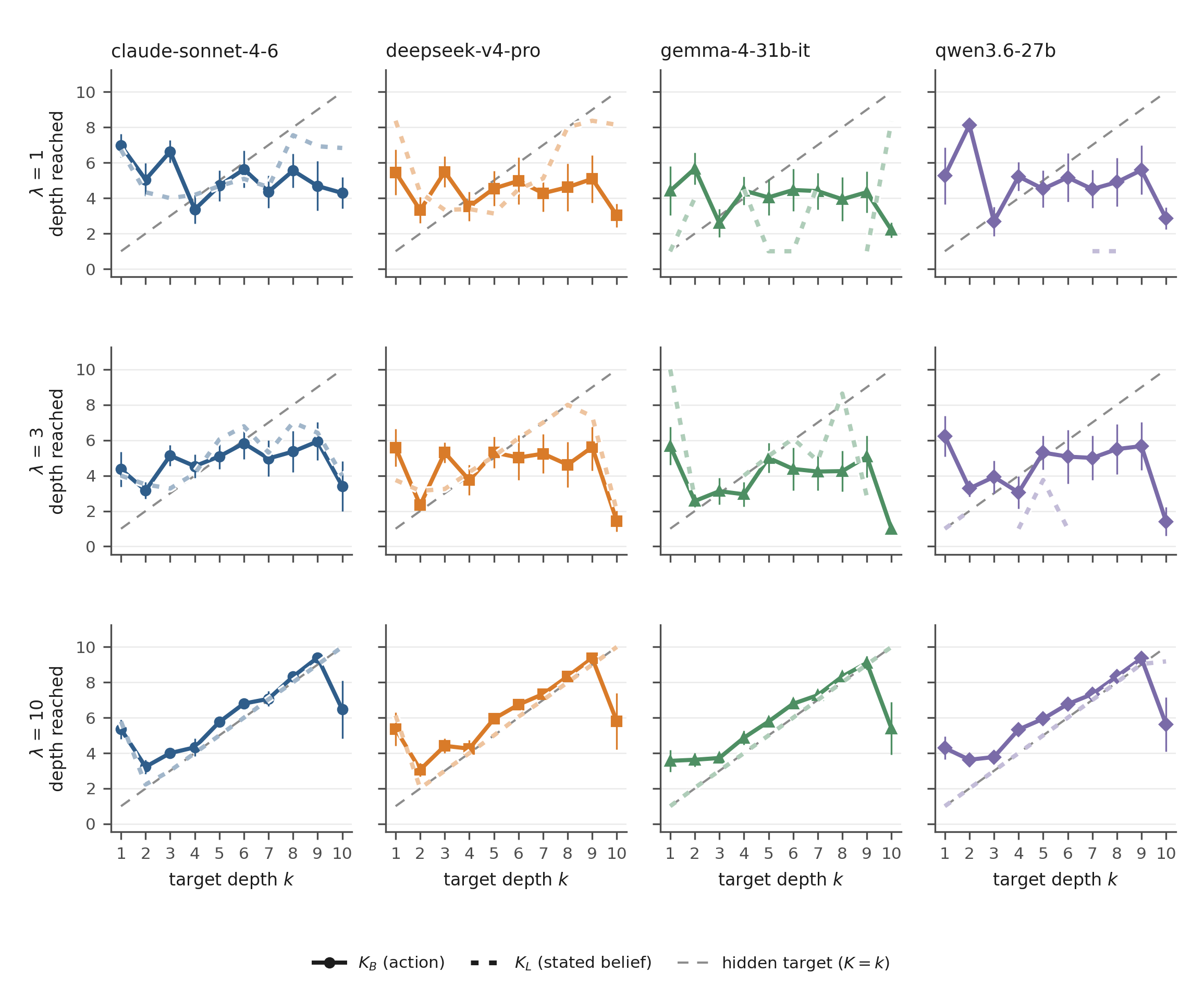}
\caption{Inductive condition: Depth limit---linguistic and behavioral signals as a function of target depth. Columns are models and rows are the sampling precision $\lambda$; solid lines and markers show $K_B$ (behavioral depth), dashed lines show $K_L$ (linguistic depth), and the gray dashed diagonal marks where $K$ equals target $k$.
}
\label{fig:inductive-depth}
\end{figure*}

In our experiments, we test the strategic reasoning depth of four models: Gemma 4 31B, Qwen 3.6 27B, Claude Sonnet 4.6, and DeepSeek v4 Pro, at temperature 0.2 with reasoning mode enabled where supported. Each model plays four games---the 11-20 Game, the All-Pay Auction, the Nash Demand Game, and Ring 11-20---against opponents at levels $L_0$--$L_9$ ($k_{\text{target}}=1$--10). The recursive condition encodes the opponent's reasoning depth in language; the inductive condition removes this description entirely and supplies only sampled historical actions, from which the model must infer the opponent's level. The language in CoT is evaluated by a blinded judge (Claude Opus 4.8) that extracts stated values without seeing the opponent description, target level, or submitted action; correctness, iterated depth response sequence (IDR), and $K_L$/$K_B$ signals are then computed deterministically from the extraction.

\providecommand{\rung}[1]{{\scriptsize(#1)}}
\begin{table*}[ht]
\small
\centering
\caption{Accuracy of reaching target strategic depth across games and condition.
Parenthesised: within-one-level accuracy, $|K_B - k_{\text{target}}| \le 1$. Rec.\ = recursive condition;
Ind.\ = inductive condition at $\lambda=10$, the sampling precision at which
the opponent's level is reliably recoverable from the observed history. Full $\lambda$ sweep in Appendix~\ref{app:lambda-sweep}.}
\label{tab:results-by-model-game}
\begin{tabular}{l cc cc cc cc}
\toprule
& \multicolumn{2}{c}{11--20 (Mod.)} & \multicolumn{2}{c}{All-Pay Auction} & \multicolumn{2}{c}{Nash Demand} & \multicolumn{2}{c}{Ring 11--20} \\
\cmidrule(lr){2-3} \cmidrule(lr){4-5} \cmidrule(lr){6-7} \cmidrule(lr){8-9}
Model & Rec. & Ind. & Rec. & Ind. & Rec. & Ind. & Rec. & Ind. \\
\midrule
Sonnet 4.6        & 49.0 \rung{99.0} & 81.0 \rung{81.0} & 54.0 \rung{100.0} & 14.0 \rung{86.0} & 49.0 \rung{91.0} & 2.0 \rung{6.0} & 41.1 \rung{97.8} & 61.1 \rung{61.1} \\
Qwen 3.6 27B      & 91.0 \rung{95.5} & 90.5 \rung{90.5} & 92.5 \rung{100.0} & 10.0 \rung{80.0} & 83.5 \rung{94.0} & 1.0 \rung{1.0} & 86.0 \rung{98.3} & 64.4 \rung{64.4} \\
Gemma 4 31B       & 88.5 \rung{96.5} & 90.0 \rung{90.0} & 87.5 \rung{95.5}  & 14.0 \rung{77.5} & 88.0 \rung{89.5} & 2.5 \rung{4.0} & 71.1 \rung{80.6} & 65.6 \rung{71.1} \\
DeepSeek v4 Pro   & 89.5 \rung{99.5} & 80.5 \rung{90.5} & 91.0 \rung{100.0} & 18.0 \rung{89.5} & 86.5 \rung{89.0} & 0.5 \rung{0.5} & 86.7 \rung{89.4} & 55.0 \rung{56.1} \\
\bottomrule
\end{tabular}
\end{table*}
\normalsize

\subsection{Recursive condition}
\label{sec:depth-ceiling}

The \textit{target level} $k_{\text{target}}$ is the depth a model must reason at to play the trial correctly. It is one level above the opponent it faces: against a level-$j$ opponent, correct play is to best-respond to $a^{(j)}$, which is level $j+1$ play, so $k_{\text{target}} = j+1$. Opponents run from $L_0$ to $L_9$, giving $k_{\text{target}} = 1$ to 10. The single correct action for a trial is therefore $a^{(k_{\text{target}})}$, the entry at that position in the game's IDR sequence, and a trial is scored accurate when the submitted action equals it. For the recursive condition, we provide information about the opponent's level both through natural language of how sophisticated they are (levels 1-4) and direct reference to their levels (higher than level 4).

Figure~\ref{fig:depth-ceiling-fig} shows how two depth signals track the target level. The behavioral signal $K_B$ is the position of the submitted action in the game's IDR sequence. The linguistic signal $K_L$ is the position of the opponent action the CoT says it expects, plus one: a model that expects a level-$j$ action and best-responds to it is itself reasoning at level $j+1$. Both signals are on the same scale as $k_{\text{target}}$, so a model that reasons to the required depth has $K_B = K_L = k_{\text{target}}$. Each signal is defined only when the value it scores matches exactly one position in the IDR sequence.

We observe \textit{internal consistency} between intermediate reasoning and action for almost all models in Figure~\ref{fig:depth-ceiling-fig}. $K_L$ and $K_B$ are tightly coupled: the dashed curves ($K_L$) and solid curves ($K_B$) are nearly indistinguishable throughout the full range $k=1$ --10. That is, we see neither $K_L > K_B$, where a model states a chain deeper than its action is carried out, nor $K_B > K_L$, where a model submits an action at a depth its CoT never states. \\
The \textit{accuracy} in choosing actions consistent with increasing the depth of the target of strategic reasoning is also high in the models. Sonnet 4.6 demonstrates behaviour that is one level lower than the target beyond level 4, and gemma-4 shows a tendency to lose the capacity at a deeper level (9 and 10).

\subsection{Inductive condition}
The inductive condition removes the opponent description and replaces it with a tool call that returns the opponent's historical trace; Figure~\ref{fig:inductive-depth} shows how the two depth signals behave under this condition. Those 30 actions are drawn by quantal best response against $a^{(k-1)}$ with precision $\lambda \in \{1, 3, 10\}$; $\lambda$ therefore determines how much information about the opponent's level the trace carries. When $\lambda=1$ and $\lambda=3$, the randomness in the trace caused decrease in accuracy, and $K_L$ does not align with $K_B$ at any depth. When $\lambda=10$, both signals track $k$ closely everywhere except $k=1$ and $k=10$, and the four models are alike in doing so. Over that same range the two signals also track each other, coming apart only at those two endpoints. At $k=1$ the opponent is non-strategic and a single best response is the whole task, but the models iterate several steps past it, so the action goes well beyond the correct level, even though the model's stated belief identifies it accurately. At $k=10$ the models identify the opponent almost exactly, yet the correct action is an extreme of the strategy space---for example, 50 (the smallest possible request) in the modified 11--20 game---but in 64 out of 80 trials the models fall back to the Level-0 default of 70 instead, so $K_B$ ends up far below the model's stated belief, which has already correctly reached the target.

\subsection{Game specific heterogeneity}

Table~\ref{tab:results-by-model-game} reports accuracy in four games and both conditions, each cell pooling the ten target depths $k_{\text{target}}=1$ -- $10$. In parentheses, we also report the accuracy by relaxing the metric such that the model hits the strategic depth at one level higher or lower for comparative analysis. A larger difference between the numbers indicates that the model came close to the exact target depth by one. In the recursive condition the four games behave alike, with only Sonnet standing well below the other three. Allowing for off-by-one error puts every model above 80\% in every game, and Sonnet wins the most of the four because it consistently plays one level short of the target (Figure~\ref{fig:depth-ceiling-fig}). In the inductive condition, accuracy remains high in the 11-20 Game, falls to roughly 55--66\% in Ring 11-20, to 10--18\% in the All-Pay Auction, and to under 3\% in the Nash Demand Game. Relaxing the accuracy criterion barely improves performance in the 11-20 Game and Ring 11-20; the All-Pay Auction shows the largest gain of any cell in the table, while the Nash Demand Game shows none. We analyze the nature of these errors more in the next section.

\begin{figure*}[!t]
\centering
\includegraphics[width=0.78\textwidth]{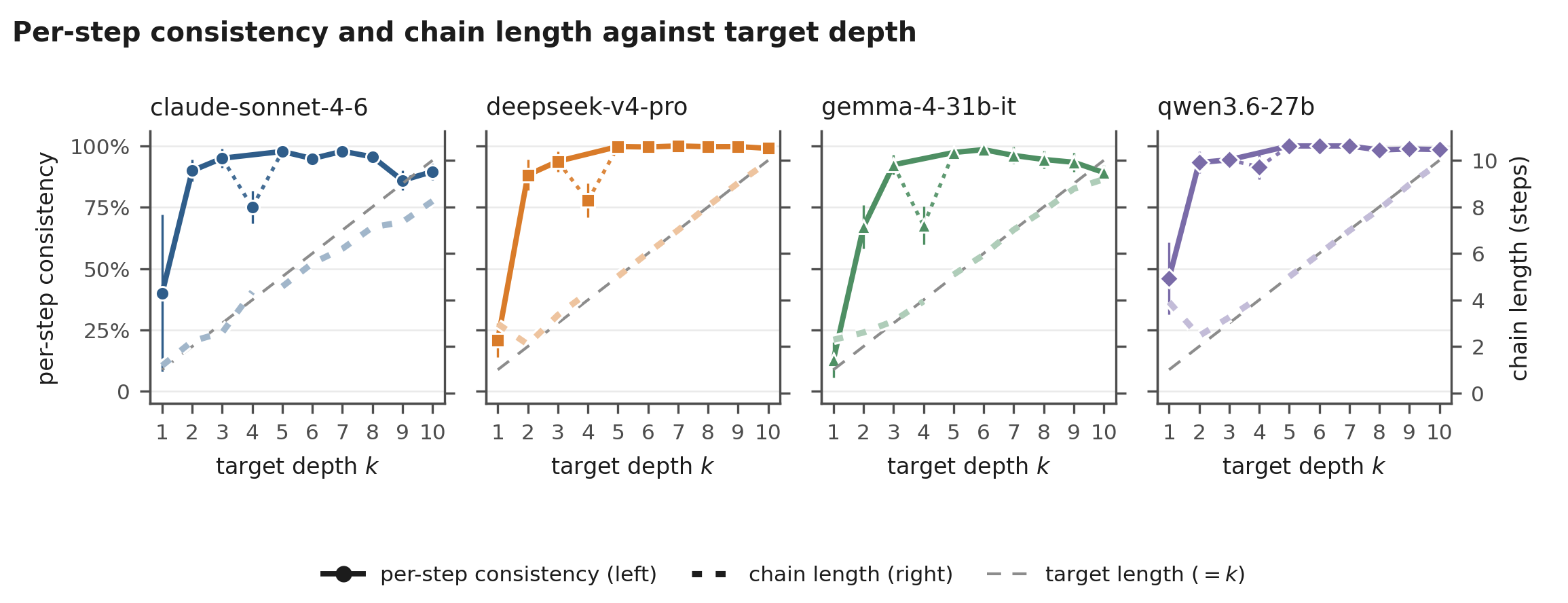}
\caption{Per-step best-response consistency and chain length against target depth $k$, over recursive-condition trials, one panel per model. The solid line with markers is per-step consistency on the left axis---the share of adjacent steps in a model's written run that are exact best responses---with 95\% confidence intervals. The heavy dotted line is mean chain length on the right axis, and the grey dashed diagonal is the target length, $m = k$.
\emph{Plotting convention:} both series break between $k=4$ and $k=5$, where the form of the opponent description changes: chain length is drawn as two separate segments, and the consistency line runs from $k=3$ straight to $k=5$, with $k=4$ hung off it on fine dotted spurs.}
\label{fig:chain-breakdown}
\end{figure*}

\subsection{Recursive Reasoning Sequence Consistency}
\label{sec:chain-length-accuracy}

Next, we analyse \textit{why} a model fails to reason at a higher strategic depth; as is the case for Sonnet-4.6. One possible explanation for this failure is an error drift: the model computes one best response incorrectly at a particular depth, and since that error persists, the chain of recursive reasoning ends on an action that is not on the correct IDR sequence. 
To test this, we let the model recursively write out a run of actions, one for each level it simulates---in the 11--20 game, 70 for the level-0 opponent, then 68 for level 1, then 66 for level 2. Let $c = (c_1, \dots, c_m)$ be that run as the blinded judge extracts it from the CoT. Based on this data, we record the following: The \textbf{length}, $m$, of the sequence, and its \textbf{per-step correctness}, i.e., the share of adjacent pairs for which $c_{t+1} \in \operatorname{BR}(c_t)$. To reach the target strategic depth, both needs to be correct. We show this analysis in Figure~\ref{fig:chain-breakdown}; the grey diagonal in Figure~\ref{fig:chain-breakdown} is the correct target length and the solid line shows the proportion of per-step correctness in best response. For the model (Sonnet-4.6) that has errors in reaching the correct strategic depth, we find that the errors arise from the incorrect \emph{length} of the IDR sequence chain it needs to recursively best respond, rather than the best response computation itself.

\subsection{Strategic Theory-of-mind in Inductive condition}
\label{sec:spurious-recursion}

Under the inductive condition, the model only has access to the samples of the opponents' gameplay. In order to infer the best response of the opponent from those samples, the model needs to apply a strategic theory of mind, to, first infer what the opponent beliefs about other players. These inferences are captured in the intermediate CoT through statements of the form ``the opponent expects me to choose X, so\ldots''. We analyze the rate of these \textit{explicit} strategic mentalizing in language under the inductive condition in Table \ref{tab:spurious-recursion}. We see that models that perform better with regard to strategic depth have a higher proportion of explicit mentalizing. When we compare the accuracy (with respect to achieving the target depth of reasoning), we observe an improvement when the model uses explicit mentalizing in its intermediate CoT.

\begin{table}[!t]
    \centering
    \small
    \setlength{\tabcolsep}{4pt}
    \caption{Strategic theory-of-mind mentalizing under the inductive
    condition. \emph{Mentalizing} is the share of trials whose chain of thought
    states that the opponent reasons about the model's own choice; the two
    accuracy columns split those trials from the rest.}
    \label{tab:spurious-recursion}
    \begin{tabular}{@{}l r r r@{}}
        \toprule
        \multirow{2}{*}{Model} & \multirow{2}{*}{Mentalizing}
          & \multicolumn{2}{c}{Accuracy} \\
        \cmidrule(lr){3-4}
        & & with & without \\
        \midrule
        Sonnet 4.6 & 6.0\% & 40.3\% & 19.3\% \\
        Qwen 3.6 27B & 28.3\% & 29.6\% & 12.4\% \\
        Gemma 4 31B & 39.3\% & 20.6\% & 15.7\% \\
        DeepSeek v4 Pro & 65.7\% & 21.4\% & 10.4\% \\
        \bottomrule
    \end{tabular}
\end{table}

\iffalse
\begin{figure*}[!t]
\centering
\includegraphics[width=0.78\textwidth]{fig4_spurious_recursion.png}
\caption{Recursive statement percentage against target depth $k$ in both conditions. Solid lines with markers show the recursive condition, with 95\% confidence intervals; the fainter dotted line shows the inductive condition, pooled over $\lambda \in \{1,3,10\}$. As in Figure~\ref{fig:chain-breakdown}, the recursive-condition line joins $k=3$ directly to $k=5$ and draws $k=4$ on dotted spurs, because the form of the opponent description changes at $k=4$.}
\label{fig:spurious-recursion}
\end{figure*}
\fi
\section{Conclusions}
This paper addressed a validity gap in existing evaluations of strategic depth in LLMs. The canonical games used by prior work were built to probe the limits of human reasoning, and nothing guarantees that distinct depths of iterated reasoning map to distinct observable actions. We made this requirement precise as the level-k
distinguishability condition, and by direct construction built a novel five-player ring game, Ring 11--20, and adapted three canonical two-player games so that their iterated depth-response sequences remain injective over the depths we wished to test. Applying these instruments to four models---Gemma 4 31B, Qwen 3.6 27B, Claude Sonnet 4.6, and DeepSeek v4 Pro---across recursive and inductive conditioning, we found accuracy to be high across the models overall, with degradation showing in model-specific ways. Sonnet 4.6's behaviour falls one level short of the target beyond depth 4, while Gemma 4 shows a tendency to lose the capacity only at the deeper levels of 9 and 10.
Looking inside the recursive condition, we found that the models' stated beliefs about their opponent ($K_L$) and their submitted actions ($K_B$) stayed tightly coupled across almost the entire depth range, so degradation at higher depths was not a case of a model acting on a belief it never stated, or vice versa. Tracing Sonnet 4.6's failures further, through its recursive reasoning depth, showed that the error was one of incorrect inference of the length of the iterated sequence the model needed to compute, rather than the best response calculations.
As for limitations, we only tested up to a depth of 10, so we cannot draw conclusions about deeper reasoning past that. We also used a single LLM judge (Claude Opus 4.8) to read and score every chain of thought, so the results are influenced by any mistake this judge makes.

\bibliography{aaai2026}

@article{nagel1995unraveling,
  author  = {Nagel, Rosemarie},
  title   = {Unraveling in Guessing Games: An Experimental Study},
  journal = {American Economic Review},
  year    = {1995},
  volume  = {85},
  number  = {5},
  pages   = {1313--1326}
}

@article{stahl1994experimental,
  author    = {Stahl, Dale O. and Wilson, Paul W.},
  title     = {Experimental Evidence on Players' Models of Other Players},
  journal   = {Journal of Economic Behavior \& Organization},
  year      = {1994},
  volume    = {25},
  number    = {3},
  pages     = {309--327},
  publisher = {Elsevier}
}

@article{stahl1995players,
  author    = {Stahl, Dale O. and Wilson, Paul W.},
  title     = {On Players' Models of Other Players: Theory and Experimental Evidence},
  journal   = {Games and Economic Behavior},
  year      = {1995},
  volume    = {10},
  number    = {1},
  pages     = {218--254},
  publisher = {Elsevier}
}

@article{costa2001cognition,
  author  = {Costa-Gomes, Miguel and Crawford, Vincent P. and Broseta, Bruno},
  title   = {Cognition and Behavior in Normal-Form Games: An Experimental Study},
  journal = {Econometrica},
  year    = {2001},
  volume  = {69},
  number  = {5},
  pages   = {1193--1235}
}

@article{costa2006cognition,
  author  = {Costa-Gomes, Miguel A. and Crawford, Vincent P.},
  title   = {Cognition and Behavior in Two-Person Guessing Games: An Experimental Study},
  journal = {American Economic Review},
  year    = {2006},
  volume  = {96},
  number  = {5},
  pages   = {1737--1768}
}

@article{camerer2004cognitive,
  author  = {Camerer, Colin F. and Ho, Teck-Hua and Chong, Juin-Kuan},
  title   = {A Cognitive Hierarchy Model of Games},
  journal = {Quarterly Journal of Economics},
  year    = {2004},
  volume  = {119},
  number  = {3},
  pages   = {861--898}
}

@incollection{camerer2004behavioural,
  author    = {Camerer, Colin F. and Ho, Teck-Hua and Chong, Juin-Kuan},
  title     = {Behavioural Game Theory: Thinking, Learning and Teaching},
  booktitle = {Advances in Understanding Strategic Behaviour: Game Theory,
               Experiments and Bounded Rationality},
  pages     = {120--180},
  year      = {2004},
  publisher = {Springer}
}

@article{crawford2013structural,
  author  = {Crawford, Vincent P. and Costa-Gomes, Miguel A. and Iriberri, Nagore},
  title   = {Structural Models of Nonequilibrium Strategic Thinking: Theory,
             Evidence, and Applications},
  journal = {Journal of Economic Literature},
  year    = {2013},
  volume  = {51},
  number  = {1},
  pages   = {5--62}
}

@article{arad2012eleven,
  author  = {Arad, Ayala and Rubinstein, Ariel},
  title   = {The 11--20 Money Request Game: A Level-$k$ Reasoning Study},
  journal = {American Economic Review},
  year    = {2012},
  volume  = {102},
  number  = {7},
  pages   = {3561--3573}
}

@article{kneeland2015identifying,
  author  = {Kneeland, Terri},
  title   = {Identifying Higher-Order Rationality},
  journal = {Econometrica},
  year    = {2015},
  volume  = {83},
  number  = {5},
  pages   = {2065--2079}
}

@techreport{cerigioni2019higher,
  author      = {Cerigioni, Francesco and Germano, Fabrizio and
                 Rey-Biel, Pedro and Zuazo-Garin, Peio},
  title       = {Higher Orders of Rationality and the Structure of Games},
  institution = {Barcelona School of Economics},
  type        = {Working Paper},
  number      = {1120},
  year        = {2019}
}

@article{coricelli2009neural,
  author    = {Coricelli, Giorgio and Nagel, Rosemarie},
  title     = {Neural Correlates of Depth of Strategic Reasoning in Medial
               Prefrontal Cortex},
  journal   = {Proceedings of the National Academy of Sciences},
  year      = {2009},
  volume    = {106},
  number    = {23},
  pages     = {9163--9168},
  publisher = {National Academy of Sciences}
}

@article{georganas2015persistence,
  author    = {Georganas, Sotiris and Healy, Paul J. and Weber, Roberto A.},
  title     = {On the Persistence of Strategic Sophistication},
  journal   = {Journal of Economic Theory},
  year      = {2015},
  volume    = {159},
  number    = {PA},
  pages     = {369--400},
  publisher = {Elsevier}
}

@article{wright2017predicting,
  author    = {Wright, James R. and Leyton-Brown, Kevin},
  title     = {Predicting Human Behavior in Unrepeated, Simultaneous-Move Games},
  journal   = {Games and Economic Behavior},
  year      = {2017},
  volume    = {106},
  pages     = {16--37},
  publisher = {Elsevier}
}

@inproceedings{wright2010beyond,
  author    = {Wright, James R. and Leyton-Brown, Kevin},
  title     = {Beyond Equilibrium: Predicting Human Behavior in Normal-Form Games},
  booktitle = {Proceedings of the Twenty-Fourth {AAAI} Conference on Artificial
               Intelligence},
  year      = {2010},
  pages     = {901--907}
}

@inproceedings{hartford2016deep,
  author    = {Hartford, Jason S. and Wright, James R. and Leyton-Brown, Kevin},
  title     = {Deep Learning for Predicting Human Strategic Behavior},
  booktitle = {Advances in Neural Information Processing Systems},
  volume    = {29},
  year      = {2016}
}

@article{baye1996allpay,
  author  = {Baye, Michael R. and Kovenock, Dan and de Vries, Casper G.},
  title   = {The All-Pay Auction with Complete Information},
  journal = {Economic Theory},
  year    = {1996},
  volume  = {8},
  number  = {2},
  pages   = {291--305}
}

@article{nash1953two,
  author  = {Nash, John},
  title   = {Two-Person Cooperative Games},
  journal = {Econometrica},
  year    = {1953},
  volume  = {21},
  number  = {1},
  pages   = {128--140}
}

@article{mallard2012modelling,
  author    = {Mallard, Graham},
  title     = {Modelling Cognitively Bounded Rationality: An Evaluative Taxonomy},
  journal   = {Journal of Economic Surveys},
  year      = {2012},
  volume    = {26},
  number    = {4},
  pages     = {674--704},
  publisher = {Wiley}
}

@inproceedings{ng2000algorithms,
  author    = {Ng, Andrew Y. and Russell, Stuart J.},
  title     = {Algorithms for Inverse Reinforcement Learning},
  booktitle = {Proceedings of the Seventeenth International Conference on
               Machine Learning},
  year      = {2000},
  pages     = {663--670}
}

@inproceedings{duetting2024mechanism,
  author    = {D{\"u}tting, Paul and Mirrokni, Vahab and Paes Leme, Renato and
               Xu, Haifeng and Zuo, Song},
  title     = {Mechanism Design for Large Language Models},
  booktitle = {Proceedings of the {ACM} Web Conference 2024},
  pages     = {144--155},
  year      = {2024}
}

@book{borgers2015introduction,
  author    = {B{\"o}rgers, Tilman},
  title     = {An Introduction to the Theory of Mechanism Design},
  year      = {2015},
  publisher = {Oxford University Press}
}

@inproceedings{kephart2016revelation,
  author    = {Kephart, Andrew and Conitzer, Vincent},
  title     = {The Revelation Principle for Mechanism Design with Reporting Costs},
  booktitle = {Proceedings of the 2016 {ACM} Conference on Economics and Computation},
  pages     = {85--102},
  year      = {2016}
}

@techreport{horton2023large,
  author      = {Horton, John J. and Filippas, Apostolos and Manning, Benjamin S.},
  title       = {Large Language Models as Simulated Economic Agents: What Can
                 We Learn from Homo Silicus?},
  institution = {National Bureau of Economic Research},
  year        = {2023}
}

@inproceedings{fu2024drive,
  author       = {Fu, Daocheng and Li, Xin and Wen, Licheng and Dou, Min and
                  Cai, Pinlong and Shi, Botian and Qiao, Yu},
  title        = {Drive Like a Human: Rethinking Autonomous Driving with Large
                  Language Models},
  booktitle    = {2024 {IEEE/CVF} Winter Conference on Applications of Computer
                  Vision Workshops ({WACVW})},
  pages        = {910--919},
  year         = {2024},
  organization = {IEEE}
}

@inproceedings{sarkar2022generalized,
  author    = {Sarkar, Atrisha and Larson, Kate and Czarnecki, Krzysztof},
  title     = {Generalized Dynamic Cognitive Hierarchy Models for Strategic
               Driving Behavior},
  booktitle = {Proceedings of the {AAAI} Conference on Artificial Intelligence},
  volume    = {36},
  number    = {5},
  pages     = {5173--5182},
  year      = {2022}
}

@article{alon2026,
  author  = {Alon, Nitay and Barnby, Joseph M. and Sarkadi, Stefan and
             Schulz, Lion and Rosenschein, Jeffrey S. and Dayan, Peter},
  title   = {{\ensuremath{\aleph}}-{IPOMDP}: Mitigating Deception in a Cognitive
             Hierarchy with Off-Policy Counterfactual Anomaly Detection},
  journal = {Journal of Artificial Intelligence Research},
  year    = {2026},
  volume  = {85}
}

@article{jia2026llm,
  author  = {Jia, Jingru and Yuan, Zehua and Pan, Junhao and
             McNamara, Paul and Chen, Deming},
  title   = {{LLM} Strategic Reasoning: Agentic Study through Behavioral Game Theory},
  journal = {Advances in Neural Information Processing Systems},
  volume  = {38},
  pages   = {58318--58350},
  year    = {2025}
}

@inproceedings{zhang2025k,
  author    = {Zhang, Yadong and Mao, Shaoguang and Ge, Tao and Wang, Xun and
               Xia, Yan and Lan, Man and Wei, Furu},
  title     = {{K}-Level Reasoning: Establishing Higher Order Beliefs in Large
               Language Models for Strategic Reasoning},
  booktitle = {Proceedings of the 2025 Conference of the Nations of the Americas
               Chapter of the Association for Computational Linguistics: Human
               Language Technologies (Volume 1: Long Papers)},
  pages     = {7212--7234},
  year      = {2025}
}

@inproceedings{trencsenyi2025approximating,
  author       = {Trencsenyi, Vince and Mensfelt, Agnieszka and Stathis, Kostas},
  title        = {Approximating Human Strategic Reasoning with {LLM}-Enhanced
                  Recursive Reasoners Leveraging Multi-Agent Hypergames},
  booktitle    = {International Workshop on Multi-Agent Systems and Agent-Based
                  Simulation},
  pages        = {15--27},
  year         = {2025},
  organization = {Springer}
}

@inproceedings{kambhampati2026position,
  author    = {Kambhampati, Subbarao and Valmeekam, Karthik and
               Bhambri, Siddhant and Palod, Vardhan and Saldyt, Lucas Paul and
               Stechly, Kaya and Samineni, Soumya Rani and Kalwar, Durgesh and
               Biswas, Upasana},
  title     = {Position: Stop Anthropomorphizing Intermediate Tokens as
               Reasoning/Thinking Traces!},
  booktitle = {Forty-third International Conference on Machine Learning Position
               Paper Track},
  year      = {2026}
}

@inproceedings{duan2024gtbench,
  author    = {Duan, Jinhao and Zhang, Renming and Diffenderfer, James and
               Kailkhura, Bhavya and Sun, Lichao and Stengel-Eskin, Elias and
               Bansal, Mohit and Chen, Tianlong and Xu, Kaidi},
  title     = {{GTBench}: Uncovering the Strategic Reasoning Capabilities of
               {LLMs} via Game-Theoretic Evaluations},
  booktitle = {Advances in Neural Information Processing Systems},
  volume    = {37},
  pages     = {28219--28253},
  year      = {2024}
}

@article{wang2024tmgbench,
  author  = {Wang, Haochuan and Feng, Xiachong and Li, Lei and Qin, Zhanyue and
             Sui, Dianbo and Kong, Lingpeng},
  title   = {{TMGBench}: A Systematic Game Benchmark for Evaluating Strategic
             Reasoning Abilities of {LLMs}},
  journal = {arXiv preprint arXiv:2410.10479},
  year    = {2024}
}

@inproceedings{huang2025gama,
  author    = {Huang, Jen-tse and Li, Eric John and Lam, Man Ho and
               Liang, Tian and Wang, Wenxuan and Yuan, Youliang and
               Jiao, Wenxiang and Wang, Xing and Tu, Zhaopeng and Lyu, Michael R.},
  title     = {Competing Large Language Models in Multi-Agent Gaming Environments},
  booktitle = {The Thirteenth International Conference on Learning Representations},
  year      = {2025}
}

@article{liu2025chbench,
  author  = {Liu, Hongtao and Du, Zhicheng and Wang, Zihe and Shen, Weiran},
  title   = {{CHBench}: A Cognitive Hierarchy Benchmark for Evaluating Strategic
             Reasoning Capability of {LLMs}},
  journal = {arXiv preprint arXiv:2508.11944},
  year    = {2025}
}

@article{kader2024emergence,
  author  = {Kader, Gavin and Lee, Dongwoo},
  title   = {The Emergence of Strategic Reasoning of Large Language Models},
  journal = {arXiv preprint arXiv:2412.13013},
  year    = {2024}
}

@inproceedings{fan2024rational,
  author    = {Fan, Caoyun and Chen, Jindou and Jin, Yaohui and He, Hao},
  title     = {Can Large Language Models Serve as Rational Players in Game
               Theory? A Systematic Analysis},
  booktitle = {Proceedings of the {AAAI} Conference on Artificial Intelligence},
  volume    = {38},
  number    = {16},
  pages     = {17960--17967},
  year      = {2024}
}

@inproceedings{turpin2023language,
  author    = {Turpin, Miles and Michael, Julian and Perez, Ethan and
               Bowman, Samuel R.},
  title     = {Language Models Don't Always Say What They Think: Unfaithful
               Explanations in Chain-of-Thought Prompting},
  booktitle = {Advances in Neural Information Processing Systems},
  volume    = {36},
  year      = {2023}
}

@article{chen2025reasoning,
  author  = {Chen, Yanda and Benton, Joe and Radhakrishnan, Ansh and
             Uesato, Jonathan and Denison, Carson and Schulman, John and
             Somani, Arushi and Hase, Peter and Wagner, Misha and
             Roger, Fabien and Mikulik, Vlad and Bowman, Samuel R. and
             Leike, Jan and Kaplan, Jared and Perez, Ethan},
  title   = {Reasoning Models Don't Always Say What They Think},
  journal = {arXiv preprint arXiv:2505.05410},
  year    = {2025}
}

@article{barez2025chain,
  author  = {Barez, Fazl and Wu, Tung-Yu and Arcuschin, Iv{\'a}n and
             Lan, Michael and Wang, Vincent and Siegel, Noah and
             Collignon, Nicolas and Neo, Clement and Lee, Isabelle and
             Paren, Alasdair and Bibi, Adel and Trager, Robert and
             Fornasiere, Damiano and Yan, John and Elazar, Yanai and
             Bengio, Yoshua},
  title   = {Chain-of-Thought Is Not Explainability},
  journal = {alphaXiv preprint alphaXiv:2025.02v1},
  year    = {2025}
}

@article{lanham2023measuring,
  author  = {Lanham, Tamera and Chen, Anna and Radhakrishnan, Ansh and
             Steiner, Benoit and Denison, Carson and Hernandez, Danny and
             Li, Dustin and Durmus, Esin and Hubinger, Evan and
             Kernion, Jackson and Perez, Ethan},
  title   = {Measuring Faithfulness in Chain-of-Thought Reasoning},
  journal = {arXiv preprint arXiv:2307.13702},
  year    = {2023}
}

@article{nie2026equilibrium,
  author  = {Nie, Wenhua and Luo, Binhan and Meng, Zijie and
             Jang, Jyh-Shing Roger and Ma, Ching-Wen},
  title   = {Equilibrium Residuals Expose Three Regimes of Matrix-Game Strategic
             Reasoning in Language Models},
  journal = {arXiv preprint arXiv:2605.10410},
  year    = {2026}
}

@article{georgousis2026counterfactual,
  author  = {Georgousis, Dimitrios and Lymperaiou, Maria and
             Dimitriou, Angeliki and Filandrianos, Giorgos and Stamou, Giorgos},
  title   = {Evaluating Counterfactual Strategic Reasoning in Large Language Models},
  journal = {arXiv preprint arXiv:2603.19167},
  year    = {2026}
}

\providecommand{\BR}{\operatorname{BR}}
\providecommand{\br}{\operatorname{br}}
\providecommand{\indic}[1]{\mathbf{1}\!\left[#1\right]}
\newcommand{\apppart}[1]{%
    \par\medskip\noindent\textbf{#1}\par\nobreak\smallskip\noindent}

\appendix
% Body sections are unnumbered (secnumdepth 0, set in the preamble), which left
% every \ref to an appendix section blank.  Appendix sections/subsections are
% numbered (A, A.1, ...) so that those cross-references resolve.
\setcounter{secnumdepth}{2}

\section{Games, IDR Sequences, and Distinguishability}
\label{app:games}

\subsection{Notation and Conventions}
\label{app:game-defs}

\apppart{Best-response set and selection.}
The \emph{best-response set} at a belief $a$ is
\begin{equation}
\BR(a) \;=\; \Big\{\, s\in S \;:\; u(s;a) = \max_{s'\in S} u(s';a) \,\Big\}.
\label{eq:br-set}
\end{equation}
This set need not be a singleton, so Eq.~\eqref{eq:br-iteration} on its own does
not define a sequence. We therefore fix a deterministic \emph{selection}, the
smallest maximiser,
\begin{equation}
\br(a) \;=\; \min \BR(a),
\label{eq:tiebreak}
\end{equation}
and iterate it, $a^{(k)} = \br\big(a^{(k-1)}\big)$, so that exactly one action is
associated with each depth.

Eq.~\eqref{eq:tiebreak} is a convention for defining the IDR sequence, not a
rule the model is asked to follow. Best-response consistency
(Sec.~\ref{app:multi-br}) is accordingly evaluated against the full set: a model
that breaks a tie differently from Eq.~\eqref{eq:tiebreak} has still played a
best response and is recorded as consistent. In the four games the convention
never binds along the sequence itself---$\BR(a^{(k)})$ is a singleton at every
depth from $0$ to $10$---so the sequences we report do not depend on it. The
only belief at which two actions tie is $a = 10$ in Ring 11--20, which lies off
the sequence (Sec.~\ref{app:ladder-derivations}).

Throughout, $s$ denotes the action of the player being scored and $a$ the action
that this player believes the opponent will take. A player in Ring 11--20 has
two neighbours, but a level-$k$ player believes that both of them play
$a^{(k-1)}$, so the payoff again depends on a single opponent action.

\apppart{Matching a number to a depth.}
A number is recorded at depth $k$ when it equals $a^{(k)}$; matching is exact
equality over the integers. The submission tool does not constrain the returned
value to $S$, so a model can submit a non-integer or an out-of-range number;
such a submission matches no depth and is scored as incorrect.

\apppart{Depths and matching windows.}
Every trial describes an opponent at level $\ell \in \{0,\dots,9\}$.
Best-responding to that opponent is play one level above it, so the target depth
is $k_{\text{target}} = \ell + 1$, running from $1$ to $10$, and the action a
trial scores as correct is always one of $a^{(1)},\dots,a^{(10)}$.

Reading a depth off a number the model produced uses a wider window,
$a^{(0)},\dots,a^{(10)}$. The window includes the anchor so that a submission
equal to $a^{(0)}$ is recorded at depth $0$ rather than as unmatched.

Two distinctness properties are used. The first is distinguishability over the
described opponents $\{a^{(0)},\dots,a^{(9)}\}$
(Definition~\ref{def:distinguishability}), which ensures that no two of them act
alike. The second is distinctness of $\{a^{(1)},\dots,a^{(10)}\}$, which ensures
that every trial has exactly one correct action and that no action is correct
for two different trials. Both properties are verified over the full window
$\{a^{(0)},\dots,a^{(10)}\}$ in Sec.~\ref{app:distinguishability}, and both hold
in all four games with a single exception, Ring 11--20 at depth $10$. This is
what the $K_{\max}$ row of Table~\ref{tab:four-games} records: distinguishability
runs through depth $10$ in 11--20, the All-Pay Auction and Nash Demand, and
through depth $9$ in Ring 11--20.

The battery is irregular in exactly two places, and both are at depth $10$: in
Nash Demand the step to $a^{(10)}$ is produced by a different branch of the
best-response map, and in Ring 11--20 $a^{(10)}$ returns to the anchor, which is
the one violation just noted. We treat both cases in
Sec.~\ref{app:distinguishability}.

\apppart{The level-0 anchor.}
The anchor is the action a player would choose if it did not reason about the
opponent at all. Each payoff function splits into a term in the player's own
action and one or more bonus terms that depend on how that action relates to the
opponent's. A non-strategic player does not ask whether a bonus condition will
in fact be met, so we strip the bonuses of their dependence on $a$: each is
treated as granted whatever the player does. What remains is a function of $s$
alone, and the anchor is its maximiser over $S$. This is the base-payoff
maximisation that the $L_0$ descriptions state in words
(Tables~\ref{tab:app-prompts-a} and~\ref{tab:app-prompts-b}).

The construction applies uniformly to all four games and reproduces every anchor
we use: it leaves $s + 32$ in 11--20, $1.8\,s$ in Nash Demand and $s + 13$ in
Ring 11--20, each maximised at the top of the strategy space, and $18 - s$ in
the All-Pay Auction, maximised at a bid of $0$.

The All-Pay anchor also follows from the maximin criterion
$\arg\max_s \min_{a} u(s;a)$, because a player who bids $s$ and loses still pays
$s$, so the worst case is least bad at a bid of $0$. The two criteria select the
same action. We report the maximin form, since it is the one expressed by the
$L_0$ description shown to the model and the one against which the judge's
\texttt{anchor\_concept} field is scored.

\begin{table}[t]
    \centering\small
    \setlength{\tabcolsep}{5pt}
    \renewcommand{\arraystretch}{1.1}
    \caption{The best-response map of each game in closed form, checked against
    every action in $S$. Branches marked \emph{generating} are those the IDR
    sequence uses; the remaining branches return the sequence to a value it has
    already visited and therefore set $K_{\max}$.}
    \label{tab:app-games}
    \begin{tabular}{@{}l l l l@{}}
        \toprule
        Game & $\br(a)$ & Range of $a$ & Role \\
        \midrule
        \multirow{2}{*}{11--20 (Mod.)}
          & $a-2$   & $a \ge 52$          & generating \\
          & $70$    & $a \le 51$          & closes the cycle \\
        \addlinespace[3pt]
        \multirow{2}{*}{All-Pay Auction}
          & $a+1$   & $a \le 9$           & generating \\
          & $0$     & $a = 10$            & closes the cycle \\
        \addlinespace[3pt]
        \multirow{2}{*}{Nash Demand}
          & $a-1$   & $a \ge 12$          & generating (undercut) \\
          & $20-a$  & $a \le 11$          & concede; $L_{10}$ onward \\
        \addlinespace[3pt]
        \multirow{3}{*}{Ring 11--20}
          & $a-3$   & $a \ge 14$          & generating ($-3$ step) \\
          & $a+7$   & $10 \le a \le 13$   & generating ($+7$ wrap) \\
          & $20$    & $a \le 9$           & never reached \\
        \bottomrule
    \end{tabular}
\end{table}

\subsection{Best-Response Maps and IDR Sequences}
\label{app:ladder-derivations}

Table~\ref{tab:app-games} gives the best-response map of each game. Iterating
these maps from the anchors gives
\begin{itemize}\itemsep2pt
    \item 11--20: $70, 68, \dots, 50$, then back to $70$;
    \item All-Pay: $0, 1, \dots, 10$, then back to $0$;
    \item Nash Demand: $20, 19, \dots, 11$, then $9$;
    \item Ring 11--20: $20, 17, 14, 11, 18, 15, 12, 19, 16, 13$, then back
          to $20$.
\end{itemize}
Two of the four change branch within the range of depths we test, and we derive
both changes here.

\apppart{Nash Demand: the change of step at $L_{10}$.}
Consider a player who expects the opponent to demand $a$. Two replies have to be
compared. To \emph{concede} is to demand the largest amount that keeps the two
demands compatible, $s = 20-a$; this is paid in full and is worth $20-a$. To
\emph{undercut} is to demand $s = a-1$; the two demands then sum to more than
20, but the smaller of them is still paid at the reduced rate of $0.8$, so the
reply is worth $0.8(a-1)$. Undercutting is the better reply exactly when
\begin{equation}
0.8(a-1) > 20-a \iff a > \tfrac{20.8}{1.8} \approx 11.56.
\label{eq:nash-threshold}
\end{equation}
The sequence begins at 20 and falls by one at each step, so it undercuts down to
11. At $a = 11$ the inequality in Eq.~\eqref{eq:nash-threshold} fails and the
player concedes, giving $\br(11) = 9$. The step is $-2$, but it comes from a
different branch of the map rather than from a larger version of the same step.

Conceding is self-cancelling. A demand of $a$ is met by conceding $20-a$, and
conceding to \emph{that} demand gives $20-(20-a) = a$ again, so two conceding
steps return to where they started. Once the sequence reaches 11 it therefore
concedes to 9, concedes back to 11, and alternates between those two values
forever.

\apppart{Ring 11--20: the sequence as a rotation.}
Every value of the Ring IDR sequence lies between 11 and 20. On that range the
two generating branches of Table~\ref{tab:app-games} can be written as a single
formula,
\begin{equation}
\br(a) \;=\; \big((a - 14) \bmod 10\big) + 11,
\qquad a \in \{11,\dots,20\},
\label{eq:ring-mod}
\end{equation}
so the best response moves the action three places backwards around a cycle of
ten. Because 3 and 10 are coprime, repeated backward steps of three visit all
ten values before returning to the starting value, which is why
$a^{(10)} = a^{(0)}$. The upward jumps of $+7$ occur at every third step: these
are the steps at which subtracting 3 would leave the range $\{11,\dots,20\}$, so
the best response wraps around instead.

Three actions have to be compared at a wrapping step. At $a = a^{(3)} = 11$ the
primary-bonus action $s=8$ pays $8+10 = 18$, the secondary-bonus action $s=18$
pays $18+3 = 21$, and the largest action that carries no bonus, $s=20$, pays 20.
The base payoff of $s=20$ already exceeds the payoff of the primary-bonus
action, and only the secondary-bonus action exceeds that base payoff, so the
sequence jumps up to 18.

At $a = 10$, one step below the range covered by Eq.~\eqref{eq:ring-mod}, the
secondary-bonus action $s = 17$ and the top action $s = 20$ both pay 20, and
Eq.~\eqref{eq:tiebreak} selects $\br(10) = 17$. This belief is never reached
from $a^{(0)} = 20$ and therefore affects no IDR sequence. The main text states
the $+7$ branch over $11 \le a \le 13$, the values the sequence actually visits;
Table~\ref{tab:app-games} extends it down to $a = 10$ only so that the closed
form agrees with the tie-break at that unreached belief.

\subsection{Distinguishability Bounds}
\label{app:distinguishability}

\begin{lemma}
\label{lem:orbit}
Let $R = \{a^{(k)} : k \ge 0\}$ be the set of values the IDR sequence ever
reaches from $a^{(0)}$. Then $K_{\max} \le |R| - 1$.
\end{lemma}

\begin{proof}
The values $a^{(0)}, a^{(1)}, \dots, a^{(|R|)}$ number $|R|+1$ and all lie in
$R$, so by the pigeonhole principle two of them are equal. Some pair of depths
$k_1 < k_2 \le |R|$ therefore shares an action, which
Definition~\ref{def:distinguishability} forbids, so distinguishability can hold
at most through depth $|R|-1$.
\end{proof}

The bound has to be stated in terms of the values the sequence actually reaches
rather than in terms of the whole strategy space, and 11--20 shows why. The
image of its best-response map is $\{50,\dots,68\}\cup\{70\}$, which is 20 of
the 21 actions in $S$; the one unreachable action is 69, which would require an
opponent action of 71. The iteration, however, starts at $a^{(0)} = 70$ and
subtracts 2 at every generating step, so it visits only the eleven even actions
between 50 and 70, giving $|R| = 11$ rather than 20.

Lemma~\ref{lem:orbit} holds with equality in all four games: each sequence
visits every value it can reach exactly once, and the following step returns it
to a value already used. This gives $K_{\max} = 10$ for 11--20, the All-Pay
Auction and Nash Demand, and $K_{\max} = 9$ for Ring 11--20, as reported in
Table~\ref{tab:four-games}. The mechanism differs across games. In 11--20 the
action that would carry the bonus falls outside $S$, so
$\br(51) = \br(50) = 70$. In the All-Pay Auction the winning bid falls outside
$S$, so $\br(10) = 0$. In Nash Demand the switch to the conceding branch makes
the sequence alternate. In Ring 11--20 the rotation of
Eq.~\eqref{eq:ring-mod} has order 10. Three of the four sequences close by
returning to the anchor, Ring 11--20 at $a^{(10)}$ and the other two at
$a^{(11)}$; Nash Demand is the exception, returning to $a^{(9)} = 11$ rather
than cycling through the whole orbit.

\apppart{Behaviour at $k_{\text{target}} = 10$.}
The deepest condition we test is not the same kind of step in all four games. In
11--20 ($a^{(10)} = 50$) and in the All-Pay Auction ($a^{(10)} = 10$) it is one
further step of the branch that generated every earlier step. In Nash Demand the
correct action, 9, is produced by the change of branch derived above, so a model
can iterate correctly and still not reach it. In Ring 11--20 the correct action
is $a^{(10)} = 20 = a^{(0)}$, so a submission of 20 matches two depths:
$\mathcal{K}_B$ then has two elements and $K_B$ is undefined
(Sec.~\ref{app:derivations}). A trial of that kind cannot identify a depth, and
we therefore exclude Ring 11--20 at $k_{\text{target}} = 10$ from the depth
figures.
These trials are excluded from the accuracy in
Table~\ref{tab:results-by-model-game} as well, so the Ring 11--20 denominator
is 720 trials per model across all four conditions rather than 800 (713 for
Qwen 3.6 27B, which lost seven trials to API errors), and 180 rather than 200
in each of the two Ring columns of that table. The exclusion is not neutral:
the ambiguous submission 20 is scored as correct at $k_{\text{target}} = 10$,
and it is submitted in 230 of the 320 affected trials, so retaining them would
raise pooled accuracy from $31.9\%$ ($n = 12{,}467$) to $32.9\%$
($n = 12{,}787$) and raise every per-model Ring figure (Sonnet 4.6
$28.6 \to 31.2\%$, DeepSeek v4 Pro $39.4 \to 44.1\%$, Gemma 4 31B
$36.0 \to 37.2\%$, Qwen 3.6 27B $38.7 \to 44.6\%$).

\subsection{Extrapolability of the IDR Sequences}
\label{app:shortcuts}

A model can arrive at $a^{(k)}$ without performing $k$ iterations. In three of
the four games the IDR sequence is an arithmetic progression over the depths it
generates, so the sequence is determined by its first two values: a model that
computes $a^{(1)}$ and applies the step $-2$, $+1$ or $-1$ can produce any
$a^{(k)}$ up to $k = 9$ in a single operation, and in 11--20 and the All-Pay
Auction $a^{(10)}$ as well. Nash Demand is the one exception at the deepest
depth: its $a^{(10)} = 9$ is produced by the change of branch derived above,
whereas extrapolating the $-1$ step predicts 10. In these three games the distance of an action from the anchor and
the depth of that action are collinear by construction, so neither accuracy nor
the two depth signals separates iteration from extrapolation.

Ring 11--20 separates them. Its sequence also has a closed form,
Eq.~\eqref{eq:ring-mod}, but that form is modular rather than affine and cannot
be recovered from the first two values. A model that extrapolates the constant
step visible at $L_1$--$L_3$ produces $20, 17, 14, 11, 8, 5, 2, \dots$ instead of
$20, 17, 14, 11, 18, 15, 12, \dots$, and the two sequences diverge at the first
wrap, $k=4$. The extrapolated values stay inside $S$ through $k=6$ and lie off
the IDR sequence at every depth from $k=4$ onward, so an extrapolating
submission is recorded as matching no depth rather than as correct.
Ring 11--20 therefore carries an extrapolation check that the other three games
cannot, which is one reason accuracy is broken out by game in
Table~\ref{tab:results-by-model-game} rather than pooled across games, and why
the per-step analysis in Figure~\ref{fig:chain-breakdown} scores the written
intermediate values rather than the endpoint alone.

\subsection{Scoring Best-Response Consistency}
\label{app:multi-br}

Best-response consistency asks whether the action a model submitted is a best
response to the belief it stated. In the recursive condition the judge extracts
a point belief $\hat b$, and the action $\hat a$ is checked against the whole
best-response set:
\begin{equation}
\text{consistent} \iff \hat a \in \BR(\hat b),
\label{eq:faithful}
\end{equation}
where $\BR(\hat b)$ is obtained by enumerating $u(s;\hat b)$ over all $s \in S$
and taking the exact maximisers.

In the inductive condition a model frequently states a distribution over
opponent actions rather than a single action. Where the judge extracts a
non-empty \texttt{belief\_support}, consistency is scored against expected
utility under that support: $\hat a$ must maximise
$\sum_j p_j\, u(\hat a; b_j)$ over $S$, where the weights $p_j$ are taken from
the likelihood phrases the model itself used. The point-belief check of
Eq.~\eqref{eq:faithful} is retained alongside this test under a separate field.
Where no support is extracted, Eq.~\eqref{eq:faithful} is used unchanged.

\section{Prompts, Models, and Execution}
\label{app:prompts}

\subsection{System Prompts}
\label{app:system-prompts}

The system prompt is the same for every game, condition, and model.

\begin{quote}\ttfamily\small
You are a rational game-theory player. Think step by step about your strategy.
Show all your reasoning, then call the submit\_action tool with your final
chosen number. You MUST use the tool to submit --- do not just state your
answer in text.
\end{quote}

\noindent In the inductive condition one further paragraph is appended:

\begin{quote}\ttfamily\small
You also have access to a query\_opponent\_history tool that returns your
opponent's raw action choices from their previous games against other players.
You may call it to inform your decision.
\end{quote}

\subsection{Game Rules and Opponent Descriptions}
\label{app:game-rules}
\label{app:format-a}

The user turn of every trial begins with the rules of the game and, in the
recursive condition, continues with a description of the opponent.
Tables~\ref{tab:app-prompts-a} and~\ref{tab:app-prompts-b} reproduce both
verbatim. The rules text is identical across conditions and depths, so the only
material that changes between trials is what follows it.

The opponent's depth is conveyed in two ways. For $L_0$--$L_3$
($k_{\text{target}} = 1$ to $4$) it is conveyed in natural language: the
description states how sophisticated the opponent is, and the model has to count
the levels from the way the sentences are nested. For $L_4$--$L_9$
($k_{\text{target}} = 5$ to $10$) the level is named directly, using a template:

\begin{quote}\ttfamily\small
\{Your opponent is $\vert$ Your opponents are\} Level \{$k$\}. A Level k player
best-responds to a Level (k-1) player, all the way down to Level 0, which is
defined as follows: \{$L_0$ description\}
\end{quote}

The level-0 sentence substituted into the template is the same string that
appears in the $L_0$ row for that game. The two formats therefore give the model
the same starting point and differ only in how the depth above that starting
point is conveyed. The wording changes between $k_{\text{target}} = 4$ and
$k_{\text{target}} = 5$, and the target depth increases across the same
boundary, so the design as run cannot separate the effect of the wording from
the effect of the depth; comparisons that cross this boundary confound the two.

The inductive condition supplies no opponent description. The user turn contains
the rules followed by:

\begin{quote}\ttfamily\small
You may use the query\_opponent\_history tool to view your opponent's past
action choices in previous games against other players. This data may help you
anticipate their strategy.
\end{quote}

No part of this text mentions levels, iteration, best responses, or the opponent
reasoning about the model. This is the basis for the explicit-mentalizing
measure reported in the main text and in Table~\ref{tab:spurious-recursion}.

\setcounter{topnumber}{3}
\setcounter{dbltopnumber}{3}
\setcounter{totalnumber}{5}
\renewcommand{\topfraction}{0.9}
\renewcommand{\dbltopfraction}{0.9}
\renewcommand{\floatpagefraction}{0.85}
\renewcommand{\dblfloatpagefraction}{0.85}
\renewcommand{\textfraction}{0.05}
\setlength{\dblfloatsep}{6pt plus 2pt minus 2pt}
\setlength{\dbltextfloatsep}{10pt plus 2pt minus 2pt}

\begin{table*}[p]
    \centering\small
    \setlength{\tabcolsep}{3pt}
    \renewcommand{\arraystretch}{0.95}
    \caption{Rules text and $L_0$--$L_3$ opponent descriptions for 11--20
    (Mod.) and the All-Pay Auction, reproduced verbatim. The $L_0$ description
    is also substituted into the template used for $L_4$--$L_9$.}
    \label{tab:app-prompts-a}
    \begin{tabular}{@{}l l p{0.80\textwidth}@{}}
        \toprule
        Game & & Text\\
        \midrule
        \multirow{5}{*}{11--20 (Mod.)}
        & Rules & \ttfamily\scriptsize I am playing a two-player game. My opponent and I each simultaneously choose an integer from 50 to 70. My payoff equals my chosen number, plus a bonus of 24 if my number is exactly 2 less than my opponent's, or a bonus of 8 if my number is exactly 1 less than my opponent's. I want to maximize my payoff. What integer should I choose?\\
        & $L_0$ & \ttfamily\scriptsize Your opponent plays naively. They have no strategic awareness and simply try to maximize their own base payoff without considering how your choice might interact with theirs.\\
        & $L_1$ & \ttfamily\scriptsize Your opponent is a basic strategic thinker. They assume you will naively maximize your own base payoff, and they choose the action that best responds to a base-payoff maximizer.\\
        & $L_2$ & \ttfamily\scriptsize Your opponent is a moderately sophisticated thinker. They assume you are a basic strategist who best-responds to naive base-payoff maximizers, and they choose the action that best responds to that level of reasoning.\\
        & $L_3$ & \ttfamily\scriptsize Your opponent is a deeply strategic thinker. They assume you are moderately sophisticated --- someone who best-responds to basic strategists --- and they choose the action that best responds to that deeper level of reasoning.\\
        \midrule
        \multirow{5}{*}{All-Pay Auction}
        & Rules & \ttfamily\scriptsize I am playing a two-player all-pay auction. My opponent and I each simultaneously choose a bid from 0 to 10 (integers only). Both of us pay our bid regardless of the outcome. If my bid is strictly higher, I win a prize of 12. If we tie, I receive 6. If my bid is lower, I receive nothing. My payoff is: prize (if any) minus my bid. I want to maximize my payoff. What integer should I bid?\\
        & $L_0$ & \ttfamily\scriptsize Your opponent plays naively. They are very cautious and would rather not bid at all than risk paying a bid and losing. They do not think strategically about what you might bid or how to compete for the prize.\\
        & $L_1$ & \ttfamily\scriptsize Your opponent is a basic strategic thinker. They assume you are very cautious and will not bid at all to avoid any risk of losing money, and they choose the bid that best responds to a non-bidder.\\
        & $L_2$ & \ttfamily\scriptsize Your opponent is a moderately sophisticated thinker. They assume you are a basic strategist who places the minimum bid needed to beat a cautious non-bidder, and they choose the bid that best responds to that level of reasoning.\\
        & $L_3$ & \ttfamily\scriptsize Your opponent is a deeply strategic thinker. They assume you are moderately sophisticated --- someone who just barely outbids basic strategists --- and they choose the bid that best responds to that deeper level of reasoning.\\
        \bottomrule
    \end{tabular}

    \bigskip
    \caption{Rules text and $L_0$--$L_3$ opponent descriptions for Nash Demand
    and Ring 11--20, reproduced verbatim.}
    \label{tab:app-prompts-b}
    \begin{tabular}{@{}l l p{0.80\textwidth}@{}}
        \toprule
        Game & & Text\\
        \midrule
        \multirow{5}{*}{Nash Demand}
        & Rules & \ttfamily\scriptsize I am playing a two-player Nash demand game. A prize of 20 is at stake. My opponent and I each simultaneously demand an integer from 1 to 20. If our demands sum to 20 or less, I receive my demand. If our demands exceed 20 but my demand is strictly less than my opponent's, I receive 80\% of my demand. Otherwise I receive nothing. I want to maximize my payoff. What integer should I demand?\\
        & $L_0$ & \ttfamily\scriptsize Your opponent plays naively. They are greedy and try to claim as much of the prize as possible for themselves, without considering how your demand might interact with theirs.\\
        & $L_1$ & \ttfamily\scriptsize Your opponent is a basic strategic thinker. They assume you are greedy and will demand as much as possible, and they choose the demand that best responds to a greedy maximizer.\\
        & $L_2$ & \ttfamily\scriptsize Your opponent is a moderately sophisticated thinker. They assume you are a basic strategist who slightly undercuts a greedy maximizer, and they choose the demand that best responds to that level of reasoning.\\
        & $L_3$ & \ttfamily\scriptsize Your opponent is a deeply strategic thinker. They assume you are moderately sophisticated --- someone who undercuts basic strategists --- and they choose the demand that best responds to that deeper level of reasoning.\\
        \midrule
        \multirow{5}{*}{Ring 11--20}
        & Rules & \ttfamily\scriptsize I am one of 5 players arranged in a ring (so player 1 neighbors players 2 and 5, etc.). Each of us simultaneously chooses an integer from 1 to 20. My payoff equals my chosen number, plus a bonus of 10 if my number is exactly 3 less than my left neighbor's number, plus a bonus of 3 if my number is exactly 7 more than my right neighbor's number. I want to maximize my payoff. What integer should I choose?\\
        & $L_0$ & \ttfamily\scriptsize Your opponents play naively. They have no strategic awareness and simply try to maximize their own base payoff without considering how other players' choices might interact with theirs.\\
        & $L_1$ & \ttfamily\scriptsize Your opponents are basic strategic thinkers. They assume other players naively maximize their own base payoff, and they choose the action that best responds to base-payoff maximizers.\\
        & $L_2$ & \ttfamily\scriptsize Your opponents are moderately sophisticated thinkers. They assume other players are basic strategists who best-respond to naive base-payoff maximizers, and they choose the action that best responds to that level of reasoning.\\
        & $L_3$ & \ttfamily\scriptsize Your opponents are deeply strategic thinkers. They assume other players are moderately sophisticated --- those who best-respond to basic strategists --- and they choose the action that best responds to that deeper level of reasoning.\\
        \bottomrule
    \end{tabular}
\end{table*}

\subsection{Models, Decoding, and Protocol}
\label{app:tools}

Table~\ref{tab:app-execution} lists the settings under which trials were run.
Run health by model, and the breakdown by target depth, are in the released
logs.

\begin{table*}[!t]
    \centering\small
    \setlength{\tabcolsep}{5pt}
    \renewcommand{\arraystretch}{1.05}
    \caption{Models, decoding, and protocol.}
    \label{tab:app-execution}
    \begin{tabular}{@{}l p{0.74\textwidth}@{}}
        \toprule
        Item & Specification\\
        \midrule
        Models & Called through OpenRouter by slug:
          \texttt{google/gemma-4-31b-it}, \texttt{qwen/qwen3.6-27b},
          \texttt{anthropic/claude-sonnet-4-6},
          \texttt{deepseek/deepseek-v4-pro}\\
        Reasoning mode & Default effort requested for \texttt{gemma-4} and
          \texttt{qwen3.6}; no reasoning parameter sent for the other two, which
          ran at the provider default. The resolved configuration is recorded
          per trial\\
        Decoding & Temperature 0.2 and \texttt{max\_tokens} 8192 sent on every
          request, and reported as sent: a provider may ignore or clamp them\\
        Routing & Provider fallback enabled, so consecutive requests may be
          served by different upstream providers; the serving provider is
          recorded per turn. Pin a single provider to reproduce a run exactly\\
        \midrule
        Tools & \texttt{submit\_action} takes the final chosen number and is
          called once after reasoning. \texttt{query\_opponent\_history} takes
          no arguments and returns 30 integers; its one-call limit is stated in
          the description but not enforced, and the call count is recorded\\
        Turn cap & 10 turns per trial\\
        Nudge & A turn ending with no tool call and no action yet submitted
          receives: \emph{``You have not submitted an action yet. Call the
          submit\_action tool now with your final chosen number. Do not write
          any more reasoning.''}\\
        Chain of thought & Reasoning-channel and visible text from every turn
          before submission, concatenated in order; text written after
          submission is excluded\\
        \midrule
        Design & $4$ games $\times$ $10$ opponent levels $\times$
          $\{$recursive, inductive $\lambda\!=\!1,3,10\}$ $\times$
          $20$ repetitions $\times$ 4 models, or 3200 trials per model
          (3201 for Gemma 4 31B)\\
        % TODO(authors): the design is 4x10x4x20 = 3200 per model.  The extra
        % Gemma 4 31B record needs a one-clause explanation here (retry kept,
        % duplicate log row, ...) or the count should be corrected to 3200.
        Failed trials & A trial that raised an API error is recorded with the
          error and excluded; a trial that completed without an action is scored
          as incorrect rather than dropped, so accuracy cannot be raised by
          declining to answer. No model exceeded a 0.7\% no-action rate. Errors: none
          for Sonnet 4.6; 22 no-action trials
          for DeepSeek v4 Pro; 3 generation and 72 judge API errors for Gemma
          4 31B; 11 generation, 24 judge API and 3 no-action for Qwen 3.6 27B.
          No generation-side error left a trial without a result record\\
        \bottomrule
    \end{tabular}
\end{table*}

\subsection{A Worked Example}
\label{app:worked-example}

Figure~\ref{fig:app-worked-example} shows a trial end to end in the format the
harness records: the two prompt turns, the tool call, the chain of thought, the
judge's extraction, and the quantities the code derives from that extraction.
The prompt turns, the extraction schema and the derived quantities are exactly
as the pipeline produces them; the chain of thought is a constructed
illustration rather than a verbatim log, so that one figure can display a fully
consistent chain alongside every derived quantity. Verbatim transcripts,
including trials whose chains are short or contain a step that is not a best
response, are in the released logs. The example is a recursive-condition trial
in the modified 11--20 game against an $L_5$ opponent, so
$k_{\text{target}} = 6$ and the correct action is $a^{(6)} = 58$.

\begin{figure*}[t]
\centering
\begin{minipage}{0.97\textwidth}
\small
\textbf{Illustrative trial}, modified 11--20 game, recursive condition,
opponent level 5, $k_{\text{target}} = 6$, correct action $a^{(6)} = 58$.

\medskip
\textbf{System turn.}
\begin{quote}\ttfamily\scriptsize
You are a rational game-theory player. Think step by step about your strategy.
Show all your reasoning, then call the submit\_action tool with your final
chosen number. You MUST use the tool to submit --- do not just state your
answer in text.
\end{quote}

\textbf{User turn} (rules, then the templated $L_5$ description).
\begin{quote}\ttfamily\scriptsize
I am playing a two-player game. My opponent and I each simultaneously choose an
integer from 50 to 70. My payoff equals my chosen number, plus a bonus of 24 if
my number is exactly 2 less than my opponent's, or a bonus of 8 if my number is
exactly 1 less than my opponent's. I want to maximize my payoff. What integer
should I choose?

Your opponent is Level 5. A Level k player best-responds to a Level (k-1)
player, all the way down to Level 0, which is defined as follows: Your opponent
plays naively. They have no strategic awareness and simply try to maximize
their own base payoff without considering how your choice might interact with
theirs.
\end{quote}

\textbf{Chain of thought in the recorded format} (reasoning channel and visible
text before submission).
\begin{quote}\ttfamily\scriptsize
Level 0 ignores the bonuses entirely and just takes the biggest number, so L0
plays 70. Against 70, the 24 bonus needs 68, worth 92, versus 69 with the 8
bonus at 77, versus 70 alone at 70. So L1 = 68. The same comparison holds at
every step: undercutting by 2 is always worth 24 against a base difference of
2, so each level subtracts 2. L2 = 66, L3 = 64, L4 = 62, L5 = 60. My opponent
is Level 5, so they will choose 60. I should best-respond to 60, which means
58 (payoff 58 + 24 = 82), better than 59 (59 + 8 = 67) or 70 (70).
\end{quote}

\textbf{Tool call.} \texttt{submit\_action(action=58)}

\medskip
\textbf{Judge extraction} (recursive-condition schema, fields shown in
Table~\ref{tab:app-judge-schema}; quotes elided for space).
\begin{quote}\ttfamily\scriptsize
\{"anchor\_present": true, "anchor\_value": 70, "anchor\_concept":
"base\_payoff\_maximizer", "chain": [\{"value": 70, "attributed\_to":
"opponent"\}, \{"value": 68, "attributed\_to": "opponent"\}, \{"value": 66,
...\}, \{"value": 64, ...\}, \{"value": 62, ...\}, \{"value": 60,
"attributed\_to": "opponent"\}], "iteration\_verbalized": "YES",
"final\_belief\_point": 60, "final\_belief\_ambiguous": false,
"belief\_is\_distribution": false, "claimed\_br": 58,
"opponent\_models\_me": false\}
\end{quote}

\textbf{Derived quantities.}
$\hat a = 58$, and $\mathcal{K}_B = \{6\}$, so $K_B = 6$.
$\hat b = 60 = a^{(5)}$, and $\mathcal{K}_L = \{5\}$, so $K_L = 6$.
Accuracy: $\hat a = a^{(6)}$, correct.
Best-response consistency: $\BR(60) = \{58\}$ and $\hat a \in \BR(60)$,
consistent.
Chain: $c = (70, 68, 66, 64, 62, 60)$, so $m = |c| = 6 = k_{\text{target}}$,
and all five consecutive pairs are exact best responses, so $C = 1$.
\end{minipage}
\caption{A trial end to end, from the two prompt turns through to the derived
quantities. The prompt turns, extraction schema and derived quantities are as
the pipeline produces them; the chain of thought is a constructed illustration,
not a verbatim log. The released logs contain the verbatim text, the judge's
quotes, and the per-turn provider record for every trial.}
\label{fig:app-worked-example}
\end{figure*}

\section{Inductive Condition: Generating the Opponent History}
\label{app:qch}

\subsection{Sampling}
\label{app:qch-sampling}

For an opponent at level $k \ge 1$, the 30 records shown to the model are drawn
independently from a quantal best response to $a^{(k-1)}$,
\begin{equation}
P_\lambda\!\left(h = s \mid k\right) =
\frac{\exp\!\big(\lambda\,\tilde u(s; a^{(k-1)})\big)}
     {\sum_{s'\in S}\exp\!\big(\lambda\,\tilde u(s'; a^{(k-1)})\big)},
\label{eq:qch}
\end{equation}
\begin{equation*}
\tilde u(s;a) = \frac{u(s;a) - u_{\min}(a)}{u_{\max}(a) - u_{\min}(a)},
\end{equation*}
where $u_{\max}(a)$ and $u_{\min}(a)$ are the largest and smallest payoffs
available at the fixed belief $a = a^{(k-1)}$. Any constant shift of $\tilde u$
leaves Eq.~\eqref{eq:qch} unchanged, so the choice of offset is immaterial.
Rescaling in this way makes a single value of $\lambda$ comparable across games
whose payoffs differ in magnitude, and it leaves the ordering of the payoffs
unchanged, so $a^{(k)}$ is the modal draw at every $\lambda$. What $\lambda$
controls is how tightly the 30 records concentrate on that value, and therefore
how clearly the opponent's depth can be read off the history.

Eq.~\eqref{eq:qch} is the quantal-response rule of the quantal cognitive
hierarchy \citep{wright2010beyond} applied to the single action $a^{(k-1)}$,
whereas a full cognitive-hierarchy belief would average over all lower types
\citep{camerer2004cognitive}. We use the degenerate level-$k$ belief because the
recursive condition describes an opponent with exactly that belief and the two
conditions should present the same opponent; where the distinction matters below
we refer to the generator as a quantal level-$k$ generator. The sampled
histories are released with the logs.

\apppart{Why level 0 is not the same opponent in the two conditions.}
A level-0 opponent has no beliefs, so there is nothing for
Eq.~\eqref{eq:qch} to best-respond to. In the inductive condition we therefore
generate its history by drawing 30 actions uniformly at random from $S$,
following the usual quantal-response convention. In the recursive condition,
level 0 is instead described in words as a specific non-strategic action, and
the model can read the anchor $a^{(0)}$ off that description. The same nominal
level thus means \emph{plays at random} in one condition and \emph{plays
$a^{(0)}$} in the other.

Two things follow. First, $\lambda$ has no effect at $k_{\text{target}} = 1$,
because uniform draws are uniform at any value of $\lambda$; the three settings
are exact replicates at that depth, and any statement about how the two depth
signals move with $\lambda$ therefore applies to $k_{\text{target}} \ge 2$.
Second, the task itself differs between the conditions at that depth.
Recursively, the model best-responds to an action it has been told about;
inductively, it has to recognise that the opponent has no pattern at all and
decide what to do about that. Accuracy at $k_{\text{target}} = 1$ is therefore
not comparable across conditions, and cross-condition comparisons at that depth
should be read with this in mind.

\subsection{The $\lambda$ Sweep}
\label{app:lambda-sweep}

The sweep is released with the logs as \texttt{lambda\_sweep.csv}: one row per
(model, game, target depth, $\lambda$) cell, 480 rows in all, each giving
accuracy and the two depth signals $K_L$ and $K_B$. This is the breakdown that
Table~\ref{tab:results-by-model-game} refers to.

\section{Judge Protocol and Derived Metrics}
\label{app:judge}

\subsection{Protocol}
\label{app:judge-protocol}

The judge receives the rules of the game and the raw chain of thought, and
nothing else: not the opponent description, the opponent's level, the target
depth, the sampled history, the tool calls, or the submitted action. Its system
prompt restricts it to reporting what the text says, and it must supply a
verbatim quote for every value it reports.

\begin{quote}\ttfamily\small
You are an annotation judge. You are given the rules of a strategic game and a
Chain-of-Thought (CoT) transcript produced by a language model playing that
game. Your task is EXTRACTION ONLY.

\medskip
Rules:

- Extract only what is explicitly stated in the CoT text. Do not infer, guess,
fill gaps, verify correctness, or assign Level-K labels.

- Every extracted value MUST include a verbatim ``quote'' copied from the CoT.

- If information is absent, use null (and empty lists for list fields).

- Report numbers exactly as stated; preserve order of appearance.

- Output ONLY a single valid JSON object matching the requested schema. No
markdown fences, no commentary.
\end{quote}

\apppart{Quote grounding.}
Every extracted value carries a quote that the schema requires to be copied from
the transcript, which gives a mechanical check on the extraction: we test each
quote for exact substring membership in the chain of thought it was drawn from,
after whitespace normalisation. The outcome of that check is recorded for every
extracted value and released with the logs.

The judge belongs to the same model family as one of the four evaluated models.
The extraction-only design limits the room for a self-preference effect, since
the judge is asked to copy values rather than to assess quality and never sees
the submitted action or whether the trial was correct, but it does not eliminate
that possibility; the per-model quote-grounding records in the released logs are
the evidence we can offer on this point.

\subsection{Extraction Schema}
\label{app:judge-schema}

The two conditions use separate schemas, because the material available in the
chain of thought differs: a recursive-condition transcript can state an anchor
and a chain of iterated values, whereas an inductive-condition transcript can
instead make claims about the sampled history.
Table~\ref{tab:app-judge-schema} lists the fields that feed the results we
report, with boolean-and-quote pairs collapsed into a single entry. The full
schemas are in the released code.

\begin{table}[!t]
    \centering\footnotesize
    \setlength{\tabcolsep}{4pt}
    \renewcommand{\arraystretch}{1.05}
    \caption{The judge fields that feed the reported results.}
    \label{tab:app-judge-schema}
    \begin{tabular}{@{}l l p{0.58\columnwidth}@{}}
        \toprule
        Group & Cond. & Fields\\
        \midrule
        Anchor & rec.
          & \texttt{anchor\_present}, \texttt{anchor\_value},
            \texttt{anchor\_concept}\\
        Chain & rec.
          & \texttt{chain[]}, entries
            $(\texttt{value},\texttt{attributed\_to},\texttt{quote})$;
            \texttt{iteration\_verbalized}\\
        History & ind.
          & \texttt{history\_referenced}, \texttt{cited\_statistics[]},
            \texttt{noise\_acknowledged}, \texttt{claimed\_n}\\
        Belief & both
          & \texttt{final\_belief\_point} with
            \texttt{final\_belief\_ambiguous} and
            \texttt{final\_belief\_candidates};
            \texttt{belief\_is\_distribution}, \texttt{belief\_support[]}\\
        Decision & both
          & \texttt{claimed\_br}, \texttt{payoff\_computations[]}\\
        Mentalizing & rec.
          & \texttt{opponent\_models\_me}\\
        Mentalizing & ind.
          & \texttt{spurious\_recursion}, the mentalizing flag behind
            Table~\ref{tab:spurious-recursion}\\
        \bottomrule
        \multicolumn{3}{@{}p{0.97\columnwidth}@{}}{\scriptsize
        \texttt{attributed\_to} $\in$ \{opponent, opponent\_model\_of\_self,
        generic\_player, unclear\};
        \texttt{iteration\_verbalized} $\in$ \{YES, NO, STOP\_RULE\_STATED\}.}
    \end{tabular}
\end{table}

Every field asks the judge to copy a value out of the text, with the exception
of the two mentalizing fields. Those two ask for a decision: whether the chain
of thought describes the opponent as reasoning about the model's own current
choice. Both record the same annotation, which we call \emph{explicit strategic
mentalizing} throughout, and both are extracted from the chain of thought alone,
without sight of the target level or the submitted action. The released code
names them differently by condition, \texttt{opponent\_models\_me} in the
recursive condition and \texttt{spurious\_recursion} in the inductive one,
because the prompt gives such a statement a basis in one and not in the other:
recursively, the prompt states the opponent's level, and a level-$k$ opponent is
by construction one that best-responds to a belief about the other player, so
the statement is supported; inductively, no part of the prompt describes the
opponent as reasoning about anyone.

\subsection{What the Code Computes from the Extraction}
\label{app:derivations}

Let $\hat a$ be the submitted action and $\hat b$ the final belief the judge
extracted. Matching is performed over the full window
$\mathcal{W} = \{0,1,\dots,10\}$ of depths for which the IDR sequence is
defined:
\begin{align}
\mathcal{K}_B &= \{k \in \mathcal{W} : a^{(k)} = \hat a\}, \\
\mathcal{K}_L &= \{k \in \mathcal{W} : a^{(k)} = \hat b\}.
\end{align}
Each signal is defined only when the corresponding set contains exactly one
element; $K_B$ is then that element and $K_L$ is that element plus one.

Because $\mathcal{W}$ contains depth 10, $K_L$ can be as large as 11, one step
beyond the deepest condition we run. Such a trial is one in which the model
states a belief appropriate to an $L_{10}$ opponent, although the deepest
opponent we describe is $L_9$. We retain these trials as resolved values of
$K_L$ rather than discarding them, because discarding them would remove
over-iteration selectively. In Nash Demand $a^{(11)} = a^{(9)}$, so a stated
belief of 11 in that game matches depth 9 alone and yields $K_L = 10$ rather
than 12.

A signal can fail to resolve in three ways, recorded separately: the value is
present but matches no depth (\emph{unmatched}); it matches more than one depth
(\emph{ambiguous}); or the chain of thought states no such value at all
(\emph{none}). The ambiguous case arises in Ring 11--20 whenever a model submits
20, which matches both $a^{(0)}$ and $a^{(10)}$. We compute the relationship
between $K_L$ and $K_B$ only over trials in which both signals resolve. Trials
with confused reasoning are both more likely to be wrong and harder to read, so
they are excluded more often than clear trials, and the agreement between the
two signals should therefore be read as an upper estimate.

The per-step measure plotted in Figure~\ref{fig:chain-breakdown} is the share of
consecutive pairs in the model's written chain $c$ for which the second value is
a best response to the first:
\begin{equation}
C = \frac{1}{|c| - 1}\sum_{t=1}^{|c|-1} \indic{c_{t+1} \in \BR(c_t)},
\qquad |c| \ge 2 .
\label{eq:chain-completeness}
\end{equation}
The membership test uses the best-response set, so a step that breaks a tie
differently from Eq.~\eqref{eq:tiebreak} is not penalised. A chain of one value
contains no consecutive pairs, so $C$ is undefined for such a chain. This
applies at $k_{\text{target}} = 1$, where a fully correct chain consists of a
single value: the $k=1$ point on that curve is computed only over the trials in
which the model wrote a second value that the task did not require, and it is
therefore not comparable with the points at greater depths. The second measure
is the chain length $m = |c|$, compared against the length the trial requires,
$m = k_{\text{target}}$.

\subsection{Uncertainty}
\label{app:ci}

For proportions we report Wilson score intervals,
\begin{equation}
\frac{\hat p + \frac{z^2}{2n} \pm z\sqrt{\frac{\hat p(1-\hat p)}{n} + \frac{z^2}{4n^2}}}{1 + \frac{z^2}{n}},
\quad z = 1.96 ,
\label{eq:wilson}
\end{equation}
rather than Wald intervals, because several of the reported quantities lie close
to 0 or 1, where the Wald interval behaves badly. For means of non-binary
quantities, such as the mean chain length in
Figure~\ref{fig:chain-breakdown}, we use normal-approximation intervals.

Repeated trials within one condition share a prompt, a sampled history and a
game, so intervals computed as if the trials were independent are too narrow. We
therefore also report a cluster bootstrap, resampling whole
(game, condition, depth, $\lambda$) cells with replacement over $B = 2000$
replicates. The resulting intervals are substantially wider than
Eq.~\eqref{eq:wilson}: a median of $4.1\times$ across the four per-model
estimates, and $7.6\times$ for the pooled estimate, where clustering bites
hardest because the pooled sample spans the most cells. Where the two disagree,
the figures show the wider interval.

\end{document}